\documentclass{article}

\usepackage{arxiv}

\usepackage[english]{babel}
\usepackage{amsthm}

\usepackage[numbers]{natbib}
\usepackage{diagbox}
\usepackage{adjustbox}
\usepackage[section]{placeins}
\usepackage{graphicx}
\usepackage[hyphens]{url}
\usepackage[utf8]{inputenc}
\usepackage{amsmath}
\usepackage{amssymb}
\usepackage{xargs}
\usepackage[pdftex,dvipsnames]{xcolor}
\usepackage{datetime}
\usepackage{verbatim}
\usepackage{listings}
\usepackage{fancyhdr}
\usepackage{fancyvrb}
\usepackage{verbatimbox}
\usepackage{latexsym}
\usepackage{nth}
\usepackage{graphicx}
\usepackage[normalem]{ulem}
\useunder{\uline}{\ul}{}
\usepackage[colorinlistoftodos,prependcaption,textsize=tiny]{todonotes}
\newcommandx{\unsure}[2][1=]{\todo[linecolor=red,backgroundcolor=red!25,bordercolor=red,#1]{#2}}
\usepackage{amsmath}

\theoremstyle{definition}

\theoremstyle{plain}
\newtheorem{proposition}{Proposition}[section]

\usepackage[fit]{truncate}

\usepackage{navigator}

\usepackage{natbib} % for bibliography
\usepackage{tabularx,longtable,multirow,subfigure,caption}%hangcaption
\usepackage{url} % URLs
\usepackage[english]{babel}
\usepackage{amsmath}
\usepackage{graphicx}
\usepackage{dsfont}
\usepackage{epstopdf} % automatically replace .eps with .pdf in graphics
\usepackage{backref} % needed for citations
\usepackage{array}
\usepackage{latexsym}
\usepackage[pdftex,pagebackref,hypertexnames=false,colorlinks]{hyperref} % provide links in pdf

\hypersetup{pdftitle={},
  pdfsubject={}, 
  pdfauthor={},
  pdfkeywords={}, 
  pdfstartview=FitH,
  pdfpagemode={UseOutlines},% None, FullScreen, UseOutlines
  bookmarksnumbered=true, bookmarksopen=true, colorlinks,
    citecolor=black,%
    filecolor=black,%
    linkcolor=black,%
    urlcolor=black}

\usepackage[all]{hypcap}

\newcommand{\R}[0]{\mathds{R}} % real numbers
\title{Applying Deep Learning to Calibrate Stochastic Volatility Models}

\author{
	\anchor{https://orcid.org/0000-0000-0000-0000}{\includegraphics[scale=0.06]{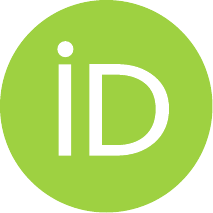}\hspace{1mm}Abir Sridi} \\
	Department of Computing \\
	Imperial College London \\
	South Kensington Campus \\
	London SW7 2AZ \\
	\texttt{a.sridi@imperial.ac.uk}
	\And
    \anchor{https://orcid.org/0000-0001-6846-6649}{\includegraphics[scale=0.06]{orcid.pdf}\hspace{1mm}Paul Bilokon} \\
	Department of Mathematics \\
	Imperial College London \\
	South Kensington Campus \\
	London SW7 2AZ \\
	\texttt{paul.bilokon@imperial.ac.uk}}

\usepackage{hyperref}
\date{Submitted: 14th September 2023, Revised: 19th September 2023}

\begin{document}

\maketitle

%%%%%%%%%%%%%%%%%%%%%%%%%%%%%%%%%%%%
\begin{abstract}

Stochastic volatility models, where the volatility is a stochastic process, can capture most of the essential stylized facts of implied volatility surfaces and give more realistic dynamics of the volatility smile/skew. However, they come with the significant issue that they take too long to calibrate. \\

Alternative calibration methods based on Deep Learning (DL) techniques have been recently used to build fast and accurate solutions to the calibration problem. Huge and Savine~\cite{differential_ML} developed a Differential Machine Learning (DML) approach, where Machine Learning models are trained on samples of not only features and labels but also differentials of labels to features. The present work aims to apply the DML technique to price vanilla European options (i.e. the calibration instruments), more specifically, puts when the underlying asset follows a Heston model and then calibrate the model on the trained network. DML allows for fast training and accurate pricing. The trained neural network dramatically reduces Heston calibration’s computation time.\\

In this work, we also introduce different regularisation techniques, and we apply them notably in the case of the DML. We compare their performance in reducing overfitting and improving the generalisation error. The DML performance is also compared to the classical DL (without differentiation) one in the case of Feed-Forward Neural Networks. We show that the DML outperforms the DL.\\

The complete code for our experiments is provided in the GitHub repository:
\begin{center}
    \url{https://github.com/asridi/DML-Calibration-Heston-Model}
\end{center}

\end{abstract}

\section{Introduction}

The development of stochastic models in exotic derivatives trading appeared with the necessity to price and hedge derivatives in financial markets as accurate as possible. The Black-Scholes model~\cite{black_scholes} is the most popular for pricing and hedging European options~\footnote{Contracts that give the buyer (in this case, the option is called Call) or seller (in this case, the option is called Put) the right but not the obligation to exercise the contract only at the maturity date at the strike price agreed at time 0.}. However, this model makes certain assumptions that can lead to prices that differ from the real-world results. In fact, this model is based on the fundamental assumption that the asset price volatility is constant over the life of the derivative. In reality, the implied volatility, defined as the volatility parameter that, when plugged into the Black–Scholes formula, enables to reproduce the market price of an option,  shows a dependence on both the strike and the option maturity.  This behaviour is known as \emph{smile} when the volatility has a minimum at the forward asset price level or \emph{skew} when low strike implied volatilities are higher than high strike ones. In the same way, models with deterministic volatility can be used to retrieve prices for a given strike exactly. However, they are unsuccessful in reproducing non-flat volatility surfaces when a whole range of strikes is considered leading to models that do not consider smile and skew effects. \\ 

The above considerations suggest the requirement for alternative models able to suitably fit the vast set of prices available to a trader. Many approaches have been suggested to tackle the issue of a good fitting of market option data. Among them, we can mention local volatility models~\cite{brigo_book, vol_local, brigo}, where the volatility is a deterministic function of both the current asset level and time. A significant drawback of this class of models is that it leads to a flattening of the forward smile, which is unrealistic and not desirable when considering exotic options that depend on the forward smile. In this case, one requires a model which gives realistic smile dynamics.\\ 

Stochastic volatility models~\cite{hull_sto, heston_vol_sto, stoch_vol_model}, where the volatility is a stochastic process, can capture most of the essential stylised facts of implied volatility surfaces and give more realistic dynamics of the volatility smile/skew. However, they come with the significant issue of taking too long to calibrate. Their calibration (i.e. the optimisation technique of finding model parameters so that the implied volatility surface produced by the model best approximates a given market implied volatility surface in a suitable metric) is time and memory intensive. This is because it performs numerical optimisation of a function that has no analytical form (this criterion function is valued numerically using Monte Carlo (MC) or Partial Differential Equation (PDE) or Fourier and asymptotic methods). Therefore, efficient calibration is prohibitively expensive, making the stochastic volatility models impractical.\\

To fill this substantial gap, alternative calibration methods based on Deep Learning (DL) techniques have been recently used to build fast and accurate solutions to the calibration problem. They can be classified into two principal approaches: 
The first approach, known as the \emph{one-step} approach, has been developed by A.Hernandez~\cite{hernandez}. The author implemented an Artificial Neural Network (ANN) that takes market data as inputs and returns calibrated model parameters as output. This approach was used to calibrate the Hull \& White model~\cite{hull_white}. The second approach, known as the \emph{two-step} approach, consists in applying the DL technique to price vanilla European options (i.e. the calibration instruments) and then calibrating the model on the trained network using the traditional optimisation methods. The latter approach has been used by Liu et al.~\cite{liu_1} to calibrate Heston and Bates models and Horvath et al.~\cite{mehdi2021deep} for the Rough Bergomi model. \\

Huge and Savine~\cite{differential_ML} developed a Differential Machine Learning (DML) approach that is supervised learning, where machine learning models are trained on samples of not only inputs and labels but also differentials of labels with respect to inputs in the context of risk management of financial Derivatives. These innovative algorithms allow fast training, accurate pricing, and risk approximations in real-time that converge. This approach applies to derivatives instruments under stochastic models, even too complex for closed-form solutions, of the underlying variables. However, applying this novel approach to stochastic volatility models is not explored. \\

The current work aims to present a data-driven framework for model calibration that allows significant computation time-saving by taking advantage of the powerful approximating capability of ANNs. Specifically, the original contributions of this study are:
\begin{itemize}
    \item Applying the DML technique to calibrate the Heston model and empirically demonstrate that the trained neural network dramatically reduces Heston calibration’s computation time and that better efficiency and accuracy in the calibration procedure are achieved.
    \item Applying different regularisation techniques to the DML and comparing their performance in reducing overfitting and improving the generalisation error.
    \item Comparing the DML performance to the classical DL (without differentiation) one to price vanilla European options when the underlying asset follows a Heston model. As a result, the DML vastly outperforms the DL.
    \item Demonstrating in the case of the Heston Model that DML could learn effectively from small datasets compared to DL. Thus training is fast (fewer labels are needed), and the pricing is accurate.\\
\end{itemize}

This work is organised as follows: In Section~\ref{chap:backgroung}, we present the calibration process in finance as well as the traditional methods to tackle the problem. Additionally, we detail how DL methods can be used to approach the calibration problem and describe the DML approach. In Section~\ref{chapter:heston_model}, we present the Heston model, describe the pricing and derive the price sensitivities to the model, option and market parameters. Next, in Section~\ref{chap:implementation_numerical_results}, we compare the DML performance to the classical DL (without differentiation).  Furthermore, we introduce different regularisation techniques, apply them notably in the case of the DML and compare their performance. Then, in
Section~\ref{chap:calibration_results}, we compare the calibration based on the neural network trained via DML (in the sequel, we name it Differential Neural Network (DNN)) to the traditional one and present how the trained neural network dramatically reduces Heston calibration’s computation time. These last two sections provide a realistic and thorough evaluation of the obtained results. Finally, in Section~\ref{chap:conclusion}, we summarise our findings and contributions and present directions for further research.

\paragraph{Acknowledgments} We would like to thank Arthur Gervais and Zhipeng Wang for their constructive comments and suggestions.

%%%%%%%%%%%%%%%%%%%%%%%%%%%%%%%%%%%%

\section{Background and Literature Review}
\label{chap:backgroung}

This work aims to take advantage of DL techniques, particularly the DML technique, to calibrate models to market data. This approach allows considerable computation time-saving. In section~\ref{section:model_calibration}, we start by presenting the calibration procedure in finance as well as traditional methods to tackle the problem; then, we detail how DL methods can be used to approach the calibration problem. Subsequently,  in section~\ref{section:FFNN}, we present the Standard Feed Forward Neural Network and detail the network training process. Next, in Section~\ref{section:DML}, we describe the DML approach.

% \subsubsection{Latin Hypercube sampling}
% \label{section:LHS}

% Latin hypercube sampling is a statistical method for generating random samples of parameter values with equal intervals.\\

% The method performs the sampling by ensuring that each sample is positioned in a space $\Omega$ of dimension ${\displaystyle d}$ as the only sample in each hyperplane of dimension ${\displaystyle d-1}$ aligned on the coordinates which define its position. Each sample is therefore positioned according to the position of the samples positioned previously to ensure that they do not have common coordinates in the space $\Omega$. This method, therefore, requires knowing the position of the samples already positioned.\\

% The Latin hypercube sampling method is entirely accepted because it presents optimisable non-collapsing space-filling properties and is flexible regarding data density and location.

\subsection{Model Calibration}
\label{section:model_calibration}

In order to compute option prices, model parameters that are not observable from the market data are required. The process of adjusting the model parameters such that model prices are compatible with market ones is called calibration. An essential issue in the calibration process is that market data is insufficient to determine the parameters accurately. Many parameter sets may produce model prices that are compatible with the market prices. In practice, it is impossible to match the market prices. Therefore the problem of calibrating the model is an optimisation problem. The calibration goal is to minimise the objective function, also called the error function between the market and the model prices for a set of parameters.\\

Various formulations for the objective function exist in the literature. In the present work, we choose to use the Weighted Root Mean Square Error (Weighted RMSE) between the model prices and the market ones with respect to the set $\Theta$ of parameters.\\

Denote $\theta \in \Theta$ the model parameters, $m = \ln(S_t/K)$ the option log-moneyness for strike $K$ and spot $S_t$, $T$ the option maturity, $\hat{V}$ and $V_{Mkt}$ respectively the model and market option prices. The weighted Root Mean Squared Error (Weighted RMSE) between the model and the market prices is given by:
\begin{equation*}
    J(\theta) = \sqrt{\sum_i \sum_j w_{i,j} \left(\hat{V}(\theta, T_i, m_j) - V_{Mkt}(T_i, m_j)\right)^2}
\end{equation*}

$w_{i,j}$ are the weights, they depend on the log-moneyness and the maturity and are chosen to reflect the importance of an option and the reliability of the market observations.

\subsubsection{Optimisation Algorithms}
\label{section:traditional_methods}

To solve the optimisation problem associated with the calibration procedure, several methods have been developed. Among them are:

\subsubsubsection{Nelder-Mead}

The Nelder-Mead method~\cite{nelder_mead} is a nonlinear optimisation algorithm published by John Nelder and Roger Mead in $1965$. It is a heuristic numerical method that seeks to minimise a continuous function in a multi-dimensional space.

Also called the downhill simplex method, the algorithm exploits the concept of simplex\footnote{In mathematics, particularly in geometry, a simplex is a triangle generalisation of the triangle to any dimension.}, a polytope\footnote{A polytope is a geometric mathematical object. A polyhedron is the best-known kind of polytope.} of $N+1$ vertices in an $N$-dimensional space. Starting initially from such a simplex, it undergoes simple transformations during the iterations: it deforms, moves and gradually shrinks until its vertices approach a point where the function is locally minimal.\\

\subsubsubsection{Differential Evolution}

Another approach, the Differential Evolution (DE)~\cite{storn}, has been developed. It is derivative-free and does not require a parameter initialisation; it allows finding a global minimum even if the objective function is non-convex. As part of this approach, Liu et al.~\cite{liu_1} described the optimisation procedure: \\

Generate $N_p$ individuals $\theta: (\theta_1, ..., \theta_{N_P}$) for calibration, add for each one a differential mutant $\theta^{'}_i = \theta_a + F . (\theta_b - \theta_c)$, where $a \neq i$, $a,b,c$ are randomly selected, $F \in [0, \infty[$ the differential weight that determines the step size of the evolution and $\theta_a$ is either a random number (strategy \emph{rand1bin}) or the best candidate of the precedent population (strategy \emph{best1bin}). Therefore, some examples are filtered out with respect to a Crossover probability $Cr$. For each individual $i$ not filtered out, the initial individual $\theta_i$ and its mutant $\theta^{'}_i$ are compared according to the objective function $J$. If the mutant $\theta^{'}_i$ returns a smaller objective function, it replaces the initial population $\theta_i.$ \\

This process is repeated until the algorithm converges or verifies a predefined criterion. The parameter selection has a significant impact on the algorithm's performance. For instance, a large mutation rate and population size may raise the probability of finding a global minimum. Another parameter, called convergence tolerance, measures the diversity within a population and determines when to finish the optimisation procedure. \\

The speed of the calibration procedure is significantly impacted by the vanilla options' pricing duration (which directly affects the objective function computation time). Inspired by the works of Liu et al.~\cite{liu_1} and Horvath et al.~\cite{mehdi2021deep}, we explore alternative methods based on ANN. They lead to a significant speed-up of the calibration procedure while preserving the accuracy of the resulting parameters.

\subsubsection{Artificial Neural Networks Based Approach}
\label{section: ANN_approach}

\subsubsubsection{Artificial Neural Networks as Function Approximators}

Consider that we would like to estimate the correct pricing function $f(x)$ by an approximate function $\hat{f}(x)$. Neural networks $F$, Defined as a function of inputs $x$ and weights $\mathbf{W}$, when trained on a dataset of inputs $x$ and outputs $y=f(x)$ are able to accurately approximate this function by finding optimal wights $\hat{\mathbf{W}}$. We can then write $F(x, \hat{\mathbf{W}}) = \hat{f}(x)$.\\

Recently, ANNs have been applied to the calibration problem. An approach pioneered by A. Hernandez~\cite{hernandez} considers the calibration as an inverse map between the market implied volatility surface and the model parameters. Hernandez applied this approach to calibrate the Hull-White model. However, even though this approach is attractive since, thanks to ANNs, it directly gives calibrated parameters using market implied volatility surface, the dependence of the network training on historical market implied volatility surfaces highlights a significant disadvantage of this approach. Indeed the availability of reliable historical market data is limited, which may lead to overfitting when the ANN is exposed to unseen data, which would happen in case of regime-switching on the market. Hence neural networks should be re-trained often due to instability, which is highly time-consuming. \\

Another approach, based on two steps, has been developed and used by Bayer et al.~\cite{mehdi_2019} and Horvath et al.~\cite{mehdi2021deep} to calibrate the Rough Bergomi model as well as Liu et al.~\cite{liu_1} for the Heston and Bates model.

\subsubsubsection{Two-step Approach as presented in~\cite{mehdi2021deep}}
\label{section:two_step_approach}

This section looks in detail at the causes of the calibration slowness. In general, thousands of evaluations of the pricing function on options with different strikes and maturities are required to find the optimal set of parameters. Nevertheless, this function is generally unknown and should be approximated via numerical methods such as MC, Fourier transformation or PDE resolution via Finite Difference. Therefore, the objective function to minimise becomes: 
\begin{equation*}
    J(\theta) = \sqrt{\sum_i \sum_j w_{i,j} \left(\hat{V}(\theta, T_i, m_j) - V_{Mkt}(T_i, m_j)\right)^2}
\end{equation*}

where $\hat{V}(\theta, T_i, k_j)$ is an estimator of the pricing function for the $i^{th}$ maturity and $j^{th}$ strike available on the market obtained using the aforementioned methods.

These numerical methods are computationally expensive and cause the slowness of the calibration process. Furthermore, they lead to numerical instabilities because of numerical approximations. Replacing this part of the calibration process with a neural network allows the minimisation to become "deterministic" and thus hugely efficient. Hence, the price computed, via numerical methods, to minimise the objective function is replaced by the one obtained by a trained ANN.\\

% \subsubsubsubsection{Approximation of the pricing map via neural networks}

\textbf{\emph{Approximation of the Pricing Map via Artificial Neural Networks}}\\

Recall that $\hat{V}$ is the model price estimator, a function of the model set of parameters $\theta$, the strike $T$ and log-moneyness $m$.\\
A neural network is trained to learn the pricing map as mentioned above. Optimal weights $\mathbf{W}$  need to be determined. Therefore, if we consider a training dataset of $N$ samples, the training is given by the following Weighted RMSE minimisation problem:

\begin{equation*}
    \hat{\mathbf{W}} := argmin_{\mathbf{W}\in\R^n} \sqrt{\sum^N_{i=1} \left(F(\mathbf{W},\theta_i,T_i,m_i)-\hat{V}(\theta_i,T_i,m_i) \right)^2}
\label{w_hat}    
\end{equation*}

Stochastic Gradient Descent~\cite{robbins_monro} or its variants such as Adam~\cite{adam}, is applied to find the optimal weights $\hat{\mathbf{W}}$. Details to follow in section~\ref{section:FFNN}.\\

% \subsubsubsubsection{Model calibration using the trained neural network}

\textbf{\emph{Model Calibration Using the Trained Neural Network}}\\

Once the optimal weights $\hat{\mathbf{W}}$ are determined, the pricing function given by the trained ANN replaces the true pricing function or its estimator, and then the calibration is executed. The model calibration problem becomes: 
\begin{equation*}
    \hat{\theta} := argmin_{\theta\in\Theta} \sqrt{\sum_{i}\sum_j w_{i,j} \left(F(\hat{\mathbf{W}}, \theta_i, T_i, m_i) - V_{Mkt}(T_i,m_j) \right)^2}
\label{w_hat_2}
\end{equation*}

A training dataset should be generated to train a neural network on the pricing map, including model, option and market parameters.\\ 

To summarise, the two-step approach that we use in this work, is as follows:

\begin{enumerate}
    \item generate a synthetic dataset for the input parameters for the model, the option and the market ones.
    \item Compute the corresponding option prices, which are considered as outputs.
    \item Split the above dataset into training and testing datasets.
    \item Train the ANN on the training dataset and evaluate it on the test dataset. 
    \item Calibrate the model on the trained network.
\end{enumerate}

This approach is stable. Indeed, using synthetic rather than historical data enables the neural networks to generate accurate results at future market states that might not have already happened. Furthermore, this method decomposes the calibration error into a pricing error (coming from the neural network) and a model miscalibration to the market data. Therefore, the neural network is trained once, making the two-step approach more robust and time-consuming.\\

Ferguson and Green~\cite{ferguson_green} noticed that the neural network gives more accurate results than MC methods. In addition, the computation of the neural network function and its gradient with respect to the model parameters is fast, accurate and efficient thanks to backpropagation, offering fast pricing and then fast calibration. This is due, as stated previously, to the fact that the optimisation problem becomes deterministic because of the ANN approximation.

\subsection{Standard Feed Forward Neural Network}
\label{section:FFNN}

Despite being the most basic architecture, Feed Forward Neural Networks (FFNNs) are considered powerful high-dimensional functional approximators. In the present study, the ANN that we use is a FFNN. This section describes standard FFNNs.\\ 

Consider a neural network with $L$ layers. Let $l \in \{1, . . . , L\}$ index the layers of the network. Let $x \in \mathbb{R}^n$ and $y \in \mathbb{R}$ denote the vector of inputs into the network and the predicted value respectively. $\mathbf{W}^{(l)} \in \mathbb{R}^{n_{l-1} \times n_l}$ and $\mathbf{b}^{(l)} \in \mathbb{R}^{n_l}$ are the weights and biases at layer $l$.
The feedforward operation of a standard neural network can be described as: for $l \in \{1, . . . , L\}$ 

\begin{align}
\label{equation:feedforward_eq}
z^{(0)} & = x  \nonumber \\
z^{(l)} & = \varphi(z^{(l-1)}) \mathbf{W}^{(l)} + \mathbf{b}^{(l)}  \\
y & = z^{(L)}  \nonumber
\end{align}

where $z^{(l)} \in \mathbb{R}^{n_l}$ is the vector containing the $n_l$ neurons, in layer $l$. The function $\varphi: \mathbb{R} \rightarrow \mathbb{R}$ is called activation function. It is typically nonlinear; otherwise, a network with more than one layer could be simplified into a single linear function. In addition, it is applied component-wise for the outputs of $\mathbf{W}^{(l)}$. The input and output layers are the first ($l=0$) and last ($l=L$) layers. Intermediate layers, $l \in \{1, . . . , L-1\}$, are called hidden layers.

\subsubsection{Initialisation}

The result of training a neural network is sensitive to the starting point. Therefore, it is necessary to initialise the network weights before starting the training. To do this, there are different techniques. By default, the weights are often initialised using a uniform or normal distribution, but more sophisticated methods exist, such as \emph{Xavier}'s or \emph{Kaiming}'s initialisation methods. On the other hand, it is often preferable to use the weights of a pre-trained network to initialise the network. These different approaches are discussed below.

\subsubsubsection{Kaiming}

Kaiming's initialisation function~\cite{kaiming} is a function that has been proposed specifically for neural networks using an activation function like ReLu~\cite{relu}. It has also been proposed as a replacement for the Xavier initialisation~\cite{xavier} for this type of network because the authors of the latter method assumed that the activation was linear, which is not the case for the ReLu. The authors of the former approach also claim that their method improves the convergence of very deep networks.

The main idea of their work was to investigate the variance of each response produced by a layer. In doing so, they concluded that a good way to initialise a layer is to use a Gaussian distribution having zero mean and standard deviation of $\sqrt{\frac{2}{n_l}}$ with $n_l$ representing the number of connections of a response.

%%%%%%%%%%%%%%%%%%%%%%%%%%%%%%%%%%%

\subsubsection{Training Standard Feed Forward Neural Network}
\label{section:train_FFNN}

Neural network training is divided into two parts, forward pass and reverse pass. In the forward pass described above, the network computes an output value for the input data. The reverse pass must now be taken, using the backpropagation technique to train the model.\\

Let us store training data in matrices:

\begin{equation*}
    X = \left[ \begin{array}{c} x_1 \\ \vdots \\ x_m  \end{array} \right] \in \mathbb{R}^{m \times n} , \ Y = \left[ \begin{array}{c} y_1 \\ \vdots \\ y_m  \end{array} \right] \in \mathbb{R}^m
\end{equation*}
where $m$ is the number of samples and $n$ is the dimension of the input.\\

During the evaluation of the network, the following matrix is computed:

\begin{equation*}
    Z^{(l)} = \left[ \begin{array}{c} z^{(l)}_1 \\ \vdots \\ z^{(l)}_m  \end{array} \right] \in \mathbb{R}^{m \times n_l}
\end{equation*}

\subsubsubsection{Cost Function}

The purpose of the network is to estimate the correct pricing function $f(x)$, which takes input $x$ and predicts the target $y$. We, therefore, want the prediction $\hat{y}$ to be as close as possible to the target $y$. In order to evaluate the performance of the model, we use for each calculated $\hat{y}$ an objective function, also called cost or loss or error function, whose value decreases when the quality of the prediction increases. That is, the lower its value, the better the model performs. Different cost functions can be used, the most common being the \emph{Mean Square Error} (MSE) which is a measure of distance between the prediction vector and the target vector:
\begin{align}
\label{equation:cost_fct}
    \mathbf{C} & = \frac{1}{m} || Y - Z^{(L)} ||^2  \nonumber \\
      & = \frac{1}{m} \left(Y - Z^{(L)}\right)^T \left(Y - Z^{(L)} \right)
\end{align}

\subsubsubsection{Gradient Descent}

Once we have a cost function, we use the gradient descent technique to train our neural network. Training a model consists of modifying the value of the connections (weights and biases) in order to obtain predictions that are closer to the targets. It is necessary to calculate the gradient of the cost according to each of the weights of the model. This calculation is used to modify the value of the weights to obtain predictions which are close to the targets.\\

It is possible to calculate the gradient for each network parameter thanks to each parameter's partial derivative. For example, assume a very simple network with a single neuron and two parameters: $w_1$ and $w_2$. Suppose also that the neuron uses an activation function $\varphi(y) = y^2$
and that the cost function is the following: $\mathbf{C}(z) = 2z$. By including all the parameters in the cost function, we obtain the following function:
\begin{equation*}
    \mathbf{C} = 2\ (w_1 x_1 + w_2 x_2)^2
\end{equation*}
It then suffices to derive this function according to each parameter to obtain the associated gradient. For example, for $w_1$, we get $\frac{\delta \mathbf{C}}{\delta w_1} = 4 x_1 (w_1 x_1 + w_2 x_2)$.\\

Computing the derivative independently for each parameter can be long and tedious, which is why the backpropagation technique is used. Indeed, this technique efficiently calculates the gradients of a neural network's weights. Backpropagation is aptly named because the flow of computations goes opposite to forward propagation. So it starts at the exit and heads towards the entrance. The idea is to calculate the derivative of the cost function with respect to the output layer and then propagate this information through the network to the model's input. The gradient of each hidden layer is expressed by reusing the derivatives of the layers that follow them.\\

Finally, to update the weights of the network,  it suffices to slightly modify the parameters in the opposite direction to the gradient according to a hyperparameter $lr$ called \emph{learning rate}, generally very small. The equation for modifying the parameter $w_1$ then becomes:
\begin{equation*}
    w_1 = w_1 - lr \frac{\delta \mathbf{C}}{\delta w_1}
\end{equation*}

\subsubsubsection{Stochastic Gradient Descent}

According to the definition of the gradient descent method, it applies after the network has seen the dataset and has calculated the value of the total cost function. However, this method neglects the information from the first training examples and is very computationally expensive.\\

In practice, the training is only done on a small, randomly chosen subset of the data. The backward pass is made only after a small subset of the data has been shown to the network. The more significant this number of data, the more stable the training because the calculated gradient more closely represents the true distribution of the training dataset. The smaller this number of data, the more random the training is, but this can prevent the network from getting stuck in a local minimum.\\

With the stochastic gradient descent, we, therefore, train on batches of data. Each time a batch is presented to the network and the parameter values are updated, it is said to come to iterate. Once all the data has been presented to the network at least once, we say that we have just completed a cycle or an \emph{epoch}.

\subsubsubsection{Adam}

One of the first improvements to gradient descent was the addition of momentum. Basically, the momentum allows to continue in the same direction as in the previous iteration, which makes it possible to speed up the optimisation. In particular, this method is very effective when the gradient direction remains relatively constant. Indeed, the longer the optimisation remains in the same direction, the greater the step to modify the parameters. Since then, several other optimisation methods have emerged. Of these, Adam~\cite{adam} is one of the most commonly used.\\

Adam is defined as follows. First, the algorithm uses first and second degree momentum. These two moments are defined as follows:

\begin{equation*}
    m_1 = \beta_1 . m_1 + (1 - \beta_1) . dx,
\end{equation*}

\begin{equation*}
    m_2 = \beta_2 . m_2 + (1 - \beta_2) . dx . dx.
\end{equation*}

$\beta_1$ and $\beta_2$ are hyperparameters usually initialised to 0.9 and 0.999, respectively. $dx$ represents the vector of the gradients calculated during the backpropagation. As $m_1$ and $m_2$ are initialised starting at zero, two terms are added to correct the bias. These terms are:
\begin{equation*}
    b_1 = \frac{m_1}{1 - \beta_1^t}
\end{equation*}

\begin{equation*}
    b_2 = \frac{m_2}{1 - \beta_2^t}
\end{equation*}
with $t$ representing the current iteration of the algorithm. Finally, the equation used by Adam to update the network weights is as follows:
\begin{equation*}
x = x - lr \frac{b_1}{\sqrt{b_2} + 10^{-8}}    
\end{equation*}

Adam is very popular because it allows a fast workout that produces good results. Moreover, it is relatively simple to implement. This is why it is often used as the default optimisation method. The only disadvantage to using Adam is that this method requires a little more memory than the others, although, in practice, this amount is often negligible and does not interfere with training.

\subsection{Differential Machine Learning}
\label{section:DML}

Huge and Savine~\cite{differential_ML} developed a Differential Machine Learning (DML) approach that is supervised learning, where Machine Learning models are trained on samples of not only inputs and labels but also differentials of labels with respect to inputs in the context of risk management of financial Derivatives. These innovative algorithms allow fast training, accurate pricing, and risk approximations in real-time that converge. Their approach applies to derivatives instruments under stochastic models, even too complex for closed-form solutions, of the underlying variables. \\

In the sequel, as in~\cite{differential_ML}, we show that  DML is an extension of classical DL applied to FFNN presented in section~\ref{section:FFNN}.

\subsubsection{Backpropagation}

We saw in section~\ref{section:FFNN} that backpropagation is a technique used to calculate the gradients of a neural network's weights. In this section, we will define backpropagation differently.\\

The differentials of the predicted value $y = z^{(L)}$ with respect to inputs $x = z^{(0)}$ are derived by differentiation of~(\ref{equation:feedforward_eq}), in the reverse order: for $l \in \{L, . . . , 1\}$ 

\begin{align}
\label{equation:backpropagation_eq}
\overline{z}^{(L)} & = \overline{y} = 1  \nonumber \\
\overline{z}^{(l-1)} & = \left(\overline{z}^{(l)} \mathbf{W}^{{(l)}^T} \right) \circ \varphi'(z^{(l-1)})   \\
\overline{x} & = \overline{z}^{(0)}  \nonumber
\end{align}

with adjoints $\overline{x} = \frac{\partial y}{\partial x}, \overline{z}_l = \frac{\partial y}{\partial z_l}, \overline{y} = \frac{\partial y}{\partial y} = 1$ and where $\circ$ is the element-wise product.\\

Backpropagation determines a feedforward neural network with inputs $\overline{y}, z^{(0)}, . . . , z^{(L)}$ and output $\overline{x} \in \mathbb{R}^n$,  where the weights are shared with the neural network defined by~(\ref{equation:feedforward_eq}) and the neurons are the adjoints of the corresponding neurons in the neural network~(\ref{equation:feedforward_eq}).

\subsubsection{Twin Network}

Huge and Savine~\cite{differential_ML} suggest concatenating equations~(\ref{equation:feedforward_eq}) and~(\ref{equation:backpropagation_eq})  into a single network representation they call \emph{twin network}.\\ 

The twin network enables the computation of a prediction jointly with its differentials with respect to inputs. The first half of the twin network is computed with the feedforward equations~(\ref{equation:feedforward_eq}) to predict a value $y$ (price). The second half is the mirror image of the first one. It is calculated with the backpropagation equations~(\ref{equation:backpropagation_eq}) to predict the differentials  $\overline{x}$ of $y$ with respect to inputs $x$. It shares weights with the first half. 
According to~(\ref{equation:backpropagation_eq}),  the neurons of the second half of the twin network are activated with the differentials $\varphi'$ of the activations $\varphi$ of the first half. Therefore, to backpropagate across the twin network, the initial activation should be $C^1 $, ruling out, for instance, \emph{ReLU} activation function. 

\subsubsection{Training as Part of the Twin Network}
\label{section:train_twin_network}

As in the case of the standard FFNN, the goal of the twin network is to estimate the correct pricing function $f(x)$, which takes input $x$ and predicts the target $y$. The difference is that the twin network learns optimal weights and biases from an augmented training set $\left(X, Y, \overline{X}\right)$, where $\overline{X}$ is the matrix of differential labels defined by:

\begin{equation*}
      \overline{X} = \left[ \begin{array}{c} \overline{x}_1 \\ \vdots \\ \overline{x}_m  \end{array} \right] \in \mathbb{R}^{m \times n} 
\end{equation*}

where $\overline{x}_i= \frac{\partial y_i}{\partial x_i}$.\\

The following matrices are computed when evaluating the twin network: $l \in \{L, . . . , 1\}$

\begin{equation*}
    Z^{(l)} = \left[ \begin{array}{c} z^{(l)}_1 \\ \vdots \\ z^{(l)}_m  \end{array} \right] \in \mathbb{R}^{m \times n_l} \ \mbox{~and~} \ \overline{Z}^{(l)} = \left[ \begin{array}{c} \overline{z}^{(l)}_1 \\ \vdots \\ \overline{z}^{(l)}_m  \end{array} \right] \in \mathbb{R}^{m \times n_l}
\end{equation*}
 
$Z^{(l)}$ (respectively $\overline{Z}^{(l)}$) is computed in the first (respectively second) half of the twin network.\\

Training with differentials $\overline{X}$ alone consists in finding weights and biases that minimise the following MSE between the label differentials and predicted differentials: 

\begin{equation}
    \label{equation:cost_fct_diff_alone}
    \mathbf{C} = \frac{1}{m}\  tr\left[\left(\overline{X} - \overline{Z}^{(0)}\right)^T \left(\overline{X} - \overline{Z}^{(0)} \right) \right]
\end{equation}
\\

When training a DNN, values and differential errors are concatenated in the cost function. Also, a penalty $\lambda$ for wrong differentials is allocated in the definition of the cost function to stop fitting noisy labels. Therefore, training a DNN consists in finding optimal weights and biases that minimise a combination between the cost functions defined by~(\ref{equation:cost_fct}) and~(\ref{equation:cost_fct_diff_alone}):

\begin{equation}
    \label{equation:cost_fct_dnn}
    \mathbf{C} = \frac{1}{m}\  \left(Y - Z^{(L)}\right)^T \left(Y - Z^{(L)} \right) + \frac{\lambda}{m} \   tr\left[\left(\overline{X} - \overline{Z}^{(0)}\right)^T \left(\overline{X} - \overline{Z}^{(0)} \right)\right] 
\end{equation}
\\

Allocating a penalty for wrong differentials allows the differential network training to stop fitting noisy labels. These criteria make label differentials act as a form of regularisation, similar to data augmentation in computer vision. The differential of labels with respect to inputs are extra training data of a different type that shares inputs with existing data; it reduces variance by augmenting the dataset size without introducing bias. Augmenting the dataset with label differentials is essential in areas where the dimension is high, and datasets are small.

\section{Application to Heston Model}
\label{chapter:heston_model}

The present study aims to apply the DML approach to price vanilla European options (i.e. the calibration instruments), specifically, puts when the underlying asset follows a Heston model~\cite{heston} and then calibrate the model using the trained network. The Heston model is one of the most popular stochastic volatility models, mainly due to its ability to produce semi-analytic formulas for plain vanilla option prices and sensitivities. In the present section, we present the model and describe the pricing under it in section~\ref{section:lit_rev_heston}. Next, we add one of our contributions: we define the Normalised Forward Put Price in section~\ref{section:nfpp} and derive semi-closed formulas for its differentials with respect to model, option and market parameters in section~\ref{section:diff_p_prime_inputs}.

\subsection{Literature Review}
\label{section:lit_rev_heston}
\subsubsection{Model Presentation}
\label{section:model_pres}

We introduce a completed filtered probability space $(\Omega, \mathcal{F}, (\mathcal{F}_t)_{t \geq 0}, \mathbb{Q})$ satisfying the usual hypothesis, where $\mathbb{Q}$ is the risk-neutral measure associated with the savings account numéraire $D_t := e^{rt}$ and $r$ is the risk-free interest rate. The Heston model is a stochastic volatility model where the dynamic of the underlying $S=\{S_t, t\geq0\}$ and its volatility $\sqrt{v_t}$ can be described by the following Stochastic Differential Equation (SDE):

\begin{equation}
\label{equation:heston_sde}
\left\{ \begin{array}{rcl}

 dS_t & = & S_t(rdt + \sqrt{v_t} dW^1_t)\\

 dv_t & = & \kappa(\theta - v_t)dt + \sigma \sqrt{v_t} dW^2_t
\end{array}\right\}
\end{equation}

where $W^1$ and $W^2$ are $\mathbb{Q}$ standard Brownian motions with:

\begin{equation*}
    d\langle W^1(t), W^2(t) \rangle = \rho dt
\end{equation*}

The instantaneous variance parameter is modelled via the CIR mean reverting process, a diffusion introduced by Cox, Ingersoll and Ross~\cite{CIR} to model short-term interest rate, which avoids null or negative values under the Feller condition, namely $2\kappa \theta > \sigma^2$. Indeed, The standard deviation term, $\sigma \sqrt v_t$, prevents negative values for all positive values of $\kappa$ and $\theta$. The zero value is also excluded if the Feller condition is satisfied.\\

The Heston instantaneous variance can be written as:
\begin{equation*}
    v_t = \theta + \exp^{-\kappa t} (v_0-\theta) + \sigma {\displaystyle \int^t_0 e^{-\kappa (t-s)} \sqrt{v_s} \,dW^2_t}
\end{equation*}

Straightforward calculations lead to:
\begin{equation*}
\left\{ \begin{array}{c}
\mathbb{E}[v_t] = e^{-\kappa t}(v_0 - \theta) + \theta\\
Var[v_t] = \frac{\sigma^2 \theta}{2\kappa} + \frac{\sigma^2 (v_0 - \theta)}{\kappa} e^{-\kappa t} + \frac{\sigma^2 (\theta -2v_0)}{2\kappa} e^{-2\kappa t}
\end{array}\right.
\end{equation*}
where $\mathbb{E}$ and $Var$ denote respectively the expectation and the variance. Particularly, for small time $t$, $Var[v_t] = v_0 \sigma^2 t$.\\

The above expressions show that the variance, initially equal to $v_0$, converges in expectation towards the long-term mean-variance $\theta$, with a speed of adjustment governed by the strictly positive parameter $\kappa$ called reversion speed or mean reversion rate. The short-term variance of $v$ indicates that its degree of stochasticity is mainly driven by $\sigma$ the volatility of variance. Finally, the correlation between the two Brownian motions shows that the spot is correlated with its volatility according to the quantity $-1<\rho<1$, which is ideally negative, signifying that a down move in the spot price is correlated with an up move in the volatility. Thus, such a dynamic can be considered economically realistic with respect to the market observations of historical volatility.

\subsubsection{Pricing under the Model}
\label{section:pricing_heston}

\subsubsubsection{Heston's Formula}
European Option prices with payoff $h(S_T)$ and maturing at time $T$, expressed as $U$ are solutions of the following PDE~\cite{liu_1}:

\begin{equation}
\label{eq:pde_heston}
    \frac{\partial U}{\partial t} + r S \  \frac{\partial U}{\partial S} + k (\theta-v) \  \frac{\partial U}{\partial v} + \frac{1}{2} v S^2 \  \frac{\partial^2 U}{\partial S^2} + \rho \sigma S v \ \frac{\partial^2 U}{\partial S \partial v} + \frac{1}{2} \sigma^2 v \ \frac{\partial^2 U}{\partial v^2} - r U = 0
\end{equation}

Note that the payoff of European call (resp. put) with strike price $K$ is $(S_T - K)_+$ (resp. $(K - S_T)_+$).\\

Under complete market and arbitrage-free conditions, using the fact that Heston model is a Markovian affine process, as demonstrated in~\cite{math_modeling_grzelak} and applying the Feynman-Kac theorem to the PDE~(\ref{eq:pde_heston}), the European option price at time $t$ is:

\begin{equation}
    U_t = \mathbb{E}[e^{-r(T-t)}h(S_T)|\mathcal{F}_t] = \mathbb{E}[e^{-r(T-t)}h(S_T)| S_t, v_t]
\end{equation}
where $\mathbb{E}$ denotes the conditional expectation.

Heston~\cite{heston} has obtained a semi-analytic formula that involves a complex integral for the valuation of European options using the characteristic function of the model. For a call, this formula is given by:

\begin{equation}
\label{equation:heston_formula}
    C_t = S P_1 - K e^{-r(T-t)} P_2
\end{equation}

where for $j = 1,2$ and $x = \ln(S)$,

\begin{equation*}
    P_j = \frac{1}{2} + \frac{1}{\pi} {\displaystyle \int^{+ \infty}_0 Re\left[\frac{e^{- i u \ln(K)} \phi_j(u)}{i u}\right]\,du}
\end{equation*}

with:

\begin{equation*}
    \phi_j(u) = \exp\left(C_j(T - t, u) + D_j(T - t, u) v + i u x \right)
\end{equation*}

where:

\begin{equation*}
    D_j(T-t, u) = \frac{b_j - \rho \sigma u i + d}{\sigma^2} \left[\frac{1 - e^{d(T-t)}}{1 - g e^{d(T-t)}} \right]
\end{equation*}

and:

\begin{equation*}
    C_j(T-t, u) = r u i (T - t) + \frac{a}{\sigma^2} \left\{(b_j - \rho \sigma u i + d )(T - t) - 2 \ln \left[\frac{1 - g e^{d(T - t)}}{1 - g} \right] \right\}
\end{equation*}

where 
\begin{equation*}
    g = \frac{b_j - \rho \sigma u i +d}{b_j - \rho \sigma u i - d}
\end{equation*} and:\\

$d = \sqrt{(\rho \sigma u i - b_j)^2 - \sigma^2(2 y_j u i - u^2)}, a = \kappa \theta, y_1 = 0.5, y_2 = -0.5, b_1 = \kappa - \rho \sigma$ and $b_2 = \kappa$.\\

\subsubsubsection{Second Formula for the Characteristic Function}

There are two forms of the Heston characteristic function $\phi$ in the literature: the original one stated above and given in Heston~\cite{heston} and the one derived in Gatheral and Taleb~\cite{gatheral_taleb_phi} that we provide below:

\begin{equation}
\label{equation:phi}
    \phi_{T-t}(u) = \exp \left[A(u) + v_0 \ B(u)\right]
\end{equation}
where:

\begin{equation*}
    A(u)= \frac{\kappa \theta}{\sigma^2}\left[(\beta-d)(T-t) - 2 \ln \left(\frac{g e^{-d(T-t)}-1}{g - 1}\right)\right]
\end{equation*}

\begin{equation*}
    B(u) = \frac{\beta-d}{\sigma^2}\left(\frac{1-e^{-d(T-t)}}{1-g e^{-d(T-t)}}\right)
\end{equation*}
with:\\

$d=\sqrt{\beta^2-4 \hat{\alpha} \gamma}$, $g=\frac{\beta-d}{\beta+d}$,
$\hat{\alpha}=-\frac{1}{2}u(u+i)$, $\beta=\kappa-iu \sigma\rho$, $\gamma=\frac{1}{2}\sigma^2$\\

Albrecher et al.~\cite{little_trap} proved that the original form is subject to numerical instabilities under certain conditions. In contrast, the stability of the second form is guaranteed under the full dimensional and unrestricted parameter space. In our work, we adopt the second form for $\phi$.
 
\subsubsubsection{Lipton's Pricing Formula}

Lewis~\cite{lewis} and Lipton~\cite{lipton2001, lipton2002} have derived formulas for the valuation problem of 
European options. Both approaches are equivalent. We provide Lipton's formulation below:

\begin{equation}
\label{equation:call_lipton_formula}
    C_t = S_t - \frac{Ke^{-r(T-t)}}{\pi} {\displaystyle \int^{+ \infty}_0 Re \left[e^{ (i u + \frac{1}{2})\hat{F}_{t,T}} \phi_{T-t}(u-\frac{i}{2})\right]\frac{\,du}{u^2+\frac{1}{4}}}
\end{equation}

where $\hat{F}_{t,T}:=\ln\left(\frac{S_t e^{r(T-t)}}{K}\right)$ denotes the log-moneyness forward. Note that $\hat{F}_{t,T} = m + r(T-t)$.\\

We prefer Equation~(\ref{equation:call_lipton_formula}) that we use in this project to~(\ref{equation:heston_formula}) provided by Heston because it has the following advantages:
\begin{itemize}
    \item less computational effort because we must calculate only one numerical integral for each price,
    \item the integrand’s denominator is quadratic in u,
    \item no need to compute the limit of the integrand near 0.
\end{itemize}

\subsection{The Normalised Forward Put Price}
\label{section:nfpp}
We define the normalised forward call price by:
\begin{equation*}
    \hat{C}_t := \frac{e^{r(T-t)}}{K} \ C_t
\end{equation*}
A straightforward calculation with the call price calculated using (\ref{equation:call_lipton_formula}) leads to:

\begin{equation*}
    \hat{C}_t = e^k - \frac{1}{\pi} {\displaystyle \int^{+ \infty}_0 Re \left[e^{ (i u + \frac{1}{2}) \hat{F}_{t,T} }\  \phi_{T-t}(u-\frac{i}{2})\right]\frac{\,du}{u^2+\frac{1}{4}}}
\end{equation*}

The call-put parity relation allows to write the normalised forward put price:
\begin{equation}
\label{equation:nfpp_def}
    \hat{P}_t:=\frac{e^{r(T-t)}}{K} \ P_t
\end{equation}
where $P_t$ is the put price at time $t$ as:

\begin{equation}
\label{equation:p_prime}
    \hat{P}_t = 1 - \frac{1}{\pi} {\displaystyle \int^{+ \infty}_0 Re \left[e^{ (i u + \frac{1}{2})\hat{F}_{t,T} \ } \phi_{T-t}(u-\frac{i}{2})\right]\frac{\,du}{u^2+\frac{1}{4}}}
\end{equation}

\subsubsection{Black-Scholes as a Particular Case}

We demonstrate that the Black-Scholes model is embedded inside the Heston model. Indeed

\begin{proposition}
\label{proposition:bs_particular_case}
For small values of $\sigma$
\begin{equation}
\label{equation:bs_particular_case}
\sigma^2_{BS} = \theta + \frac{1-e^{-\kappa  (T-t)}}{\kappa (T-t)} \  (v_0-\theta)
\end{equation}
\end{proposition}

\begin{proof}
See Appendix \ref{section:dem_bs_particular_case}.
\end{proof}

In particular, $\sigma \sim 0$ and $\theta = v_0$ produce the Black-Scholes volatility $\sigma_{BS} = \sqrt{v_0}$.

Therefore, the European option Heston price under these parameter values is the Black-Scholes price. 

%%%%%%%%%%%%%%%%%%%%%%%%%%%%%%%%%%%%
\subsection{Differentials of the Normalised Forward Put Price}
\label{section:diff_p_prime_inputs}

In this section, we derive semi-closed formulas for the differential of the normalised forward put price with respect to the long-term mean-variance $\theta$, the log-moneyness $m$, the initial variance $v_0$, the risk-free rate $r$, the time to maturity $\tau:= T-t$, the mean reversion speed of variance process $\kappa$, the correlation between the variance and the underlying process $\rho$ and the volatility of the variance $\sigma$. 

%%%%%%%%%%%%%%%%%%%%%%%%%%%%%%%%%%%%%%%%%%%%
\subsubsubsection{Partial derivative $\frac{\partial \hat{P}}{\partial \theta}$}\\
\label{section:diff_wrt_theta}

The partial derivative of the normalised forward put price $\hat{P}$ with respect to the long-term mean-variance $\theta$ is given by:

\begin{equation}
\label{equation:diff_wrt_theta}
    \frac{\partial \hat{P}}{\partial \theta} =  - \frac{1}{\pi \theta} {\displaystyle \int^{+ \infty}_0 Re \left[e^{ (i u + \frac{1}{2})\hat{F}_{t,T}} A(u-\frac{i}{2}) \phi_{\tau}(u-\frac{i}{2}) \right] \frac{\,du}{u^2+\frac{1}{4}}}
\end{equation}

\begin{proof}
See Appendix \ref{section:demo_diff_wrt_theta} .
\end{proof}

%%%%%%%%%%%%%%%%%%%%%%%%%%%%%%%%%%%%%%%%
\subsubsubsection{Partial derivative $\frac{\partial \hat{P}}{\partial m}$}
\label{section:diff_wrt_m}

The partial derivative of the normalised forward put price $\hat{P}$ with respect to  the log-moneyness $m$ is given by:

\begin{equation}
\label{equation:diff_wrt_m}
    \frac{\partial \hat{P}}{\partial m} =  - \frac{1}{\pi} {\displaystyle \int^{+ \infty}_0 Re \left[e^{ (i u + \frac{1}{2})\hat{F}_{t,T}} (i u + \frac{1}{2}) \phi_{\tau}(u-\frac{i}{2}) \right]\frac{\,du}{u^2+\frac{1}{4}}}
\end{equation}

\begin{proof}
See Appendix \ref{section:demo_diff_wrt_m}.
\end{proof}
%%%%%%%%%%%%%%%%%%%%%%%%%%%%%%%
\subsubsubsection{Partial derivative $\frac{\partial \hat{P}}{\partial v_0}$}
\label{section:diff_wrt_v0}

The partial derivative of the normalised forward put price $\hat{P}$ with respect to  the initial variance $v_0$ is given by:

\begin{equation}
\label{equation:diff_wrt_v0}
    \frac{\partial \hat{P}}{\partial v_0} =  - \frac{1}{\pi} {\displaystyle \int^{+ \infty}_0 Re \left[e^{ (i u + \frac{1}{2})\hat{F}_{t,T}} B(u-\frac{i}{2}) \phi_{\tau}(u-\frac{i}{2}) \right]\frac{\,du}{u^2+\frac{1}{4}}}
\end{equation}

\begin{proof}
See Appendix \ref{section:demo_diff_wrt_v0}.
\end{proof}
%%%%%%%%%%%%%%%%%%%%%%%%%%%%%%%%%%%%%%%%%%%%%%%%%%%%%%%%%

%%%%%%%%%%%%%%%%%%%%%%%%%%%%%%%
\subsubsubsection{Partial derivative $\frac{\partial \hat{P}}{\partial r}$}
\label{section:diff_wrt_r}

The partial derivative of the normalised forward put price $\hat{P}$ with respect to  the risk-free rate $r$ is given by:

\begin{equation}
\label{equation:diff_wrt_r}
    \frac{\partial \hat{P}}{\partial r} = \tau \frac{\partial \hat{P}}{\partial m}
\end{equation}
where:
$\frac{\partial \hat{P}}{\partial m}$ is given by (\ref{equation:diff_wrt_m})

\begin{proof}
See Appendix \ref{section:demo_diff_wrt_r}.
\end{proof}
%%%%%%%%%%%%%%%%%%%%%%%%%%%%%%%%%%%%%%%%%%%%%%%%%%%%%%%%%

%%%%%%%%%%%%%%%%%%%%%%%%%%%%%%%%%%%%%%%%
\subsubsubsection{Partial derivative $\frac{\partial \hat{P}}{\partial \tau}$}
\label{section:diff_wrt_tau}

The partial derivative of the normalised forward put price $\hat{P}$ with respect to the time to maturity $\tau$ is given by:

\begin{eqnarray}
\label{equation:diff_wrt_tau}
    \frac{\partial \hat{P}}{\partial \tau} =  - \frac{1}{\pi} \biggl\{ r {\displaystyle \int^{+ \infty}_0 Re \left[e^{ (i u + \frac{1}{2})\hat{F}_{t,T}} (i u + \frac{1}{2}) \phi_{\tau}(u-\frac{i}{2}) \right]\frac{\,du}{u^2+\frac{1}{4}}} \\
    + {\displaystyle \int^{+ \infty}_0 Re \left[e^{ (i u + \frac{1}{2})\hat{F}_{t,T}} \ \frac{\partial \phi_{\tau}(u-\frac{i}{2})}{\partial \tau}\right] \frac{\,du}{u^2+\frac{1}{4}}} \biggl\} \nonumber
\end{eqnarray}

where:

\begin{equation}
\label{equation:diff_phi_tau}
\frac{\partial \phi_{\tau}(u)}{\partial \tau} = \phi_{\tau}(u)\left[\frac{\kappa \theta}{\sigma^2} \left(\beta - d - 2 \  \frac{d g e^{-d\tau}}{1-ge^{-d\tau}} \right) \\
  + v_0 \ \frac{B(u)\ d\ e^{-d\tau} (1-g) }{(1- e^{-d\tau}) (1-g e^{-d\tau})}\right]
\end{equation}

\begin{proof}
See Appendix \ref{section:demo_diff_wrt_tau}.
\end{proof}
%%%%%%%%%%%%%%%%%%%%%%%%%%%%%%%

%%%%%%%%%%%%%%%%%%%%%%%%%%%%%%%%%%%%%%%%
\subsubsubsection{Partial derivative $\frac{\partial \hat{P}}{\partial \kappa}$}
\label{section:diff_wrt_kappa}

The partial derivative of the normalised forward put price $\hat{P}$ with respect to the mean reversion speed of variance process $\kappa$ is given by:

\begin{equation}
\label{equation:diff_wrt_kappa}
    \frac{\partial \hat{P}}{\partial \kappa} =  - \frac{1}{\pi} {\displaystyle \int^{+ \infty}_0 Re \left[ e^{ (i u + \frac{1}{2}) \hat{F}_{t,T}} \ \frac{\partial \phi_{\tau}(u-\frac{i}{2})}{\partial \kappa}\right]\frac{\,du}{u^2+\frac{1}{4}}}
\end{equation}

where:

\begin{equation}
\label{equation:diff_phi_kappa}
    \frac{\partial \phi_{\tau}(u)}{\partial \kappa} = \phi_{\tau}(u) \left[
    \frac{\partial A(u)}{\partial \kappa} + v_0 \frac{\partial B(u)}{\partial \kappa} \right]
\end{equation}

with:
\begin{equation}
\label{equation:diff_A_kappa}
\frac{\partial A(u)}{\partial \kappa} = \frac{A(u)}{\kappa} - \frac{\kappa \theta}{d \sigma^2}\left[-d \tau + \tau \beta + \frac{4g}{g-1} + \frac{2\ g\ e^{-d\tau} (2+\tau \beta) }{1-g e^{-d\tau}} \right]
\end{equation}

and:
\begin{equation}
\label{equation:diff_B_kappa}
\frac{\partial B(u)}{\partial \kappa} = \frac{B(u)}{d} \left[ - 1 + \frac{\tau \ \beta \ e^{-d\tau} }{1-e^{-d\tau}} - \frac{g \ e^{-d\tau}\ (2+\tau \beta) }{1-g\ e^{-d\tau}} \right]
\end{equation}

\begin{proof}
See Appendix \ref{section:demo_diff_wrt_kappa}.
\end{proof}
%%%%%%%%%%%%%%%%%%%%%%%%%%%%%%%
\subsubsubsection{Partial derivative $\frac{\partial \hat{P}}{\partial \rho}$}
\label{section:diff_wrt_rho}

The partial derivative of the normalised forward put price $\hat{P}$ with respect to the  correlation between the variance and the underlying process $\rho$ is given by:

\begin{equation}
\label{equation:diff_wrt_rho}
    \frac{\partial \hat{P}}{\partial \rho} =  - \frac{1}{\pi} {\displaystyle \int^{+ \infty}_0 Re \left[ e^{ (i u + \frac{1}{2}) \hat{F}_{t,T}} \ \frac{\partial \phi_{\tau}(u-\frac{i}{2})}{\partial \rho}\right]\frac{\,du}{u^2+\frac{1}{4}}}
\end{equation}

where:

\begin{equation}
\label{equation:diff_phi_rho}
     \frac{\partial \phi_{\tau}(u)}{\partial \rho} = \phi_{\tau}(u) \left[
    \frac{\partial A(u)}{\partial \rho} + v_0 \frac{\partial B(u)}{\partial \rho} \right]
\end{equation}

with:
\begin{equation}
\label{equation:diff_A_rho}
\frac{\partial A(u)}{\partial \rho} = \frac{\kappa \theta i u}{d \sigma} \left \{\tau (\beta - d) - 2 g \left[ \frac{e^{-d\tau} (2+\tau \beta)}{g e^{-d\tau} - 1} - \frac{2}{g-1} \right] \right \}
\end{equation}

and:
\begin{equation}
\label{equation:diff_B_rho}
\frac{\partial B(u)}{\partial \rho} = \frac{i u \sigma}{d} \ B(u) \left[ 1+ e^{-d\tau}\left(\frac{-\tau \beta}{1 - e^{-d\tau}} + \frac{g(2 + \tau \beta)} {1-g e^{-d\tau}} \right) \right]
\end{equation}

\begin{proof}
See Appendix \ref{section:demo_diff_wrt_rho}.
\end{proof}

%%%%%%%%%%%%%%%%%%%%%%%%%%%%%%%
%%%%%%%%%%%%%%%%%%%%%%%%%%%%%%%
\subsubsubsection{Partial derivative $\frac{\partial \hat{P}}{\partial \sigma}$}
\label{section:diff_wrt_sigma}

The partial derivative of the normalised forward put price $\hat{P}$ with respect to the volatility of the variance $\sigma$ is given by:

\begin{equation}
\label{equation:diff_wrt_sigma}
    \frac{\partial \hat{P}}{\partial \sigma} =  - \frac{1}{\pi} {\displaystyle \int^{+ \infty}_0 Re \left[ e^{ (i u + \frac{1}{2}) \hat{F}_{t,T}} \ \frac{\partial \phi_{\tau}(u-\frac{i}{2})}{\partial \sigma}\right]\frac{\,du}{u^2+\frac{1}{4}}}
\end{equation}

where:

\begin{equation}
\label{equation:diff_phi_rho_1}
     \frac{\partial \phi_{\tau}(u)}{\partial \sigma} = \phi_{\tau}(u) \left[
    \frac{\partial A(u)}{\partial \sigma} + v_0 \frac{\partial B(u)}{\partial \sigma} \right]
\end{equation}

with:
\begin{equation}
\label{equation:diff_A_sigma}
\frac{\partial A(u)}{\partial \sigma} = \frac{-2\ A(u)}{\sigma} + \frac{\kappa \theta}{\sigma^2}\left \{-i u \tau \rho + \frac{\tau}{d} (i u \rho \beta + 2 \hat{\alpha} \sigma) - 2 \ \frac{\partial}{\partial \sigma}\left[ \ln \left (\frac{g e^{-d\tau} -1} {g-1} \right) \right] \right \}
\end{equation}
where:

\begin{multline}
\label{equation:diff_ln_sigma}
\frac{\partial}{\partial \sigma} \left[ \ln \left(\frac{g e^{-d\tau} - 1}{g - 1} \right) \right] = \\
\\
\frac{\left[2 i u \rho \left(\beta^2 - d^2 \right) + 4 \beta \hat{\alpha} \sigma \right] \left[(g - 1) e^{-d\tau} - g e^{-d\tau} + 1 \right] + \tau g (g - 1) e^{-d\tau} (\beta + d)^2  (i u \rho \beta + 2 \hat{\alpha} \sigma) }{d \ (\beta + d)^2 (g - 1) \left(g e^{-d\tau} - 1 \right)}
\end{multline}

and:

\begin{equation}
\label{equation:diff_B_sigma}
\frac{\partial B(u)}{\partial \sigma} = \left( \frac{1 - e^{-d\tau}}{1 - g e^{-d\tau}} \right)  \frac{\partial}{\partial \sigma} \left( \frac{\beta-d}{\sigma^2} \right) + \frac{\beta-d}{\sigma^2} \  \frac{\partial}{\partial \sigma} \left(\frac{1 - e^{-d\tau}}{1 - g e^{-d\tau}} \right)
\end{equation}

where:

\begin{equation}
\label{equation:diff_beta_d_sigma}
\frac{\partial}{\partial \sigma} \left(\frac{\beta - d}{\sigma^2} \right) = \frac{\left[\sigma (- i u \rho d + i u \rho \beta + 2 \hat{\alpha} \sigma) - 2 d (\beta - d) \right] \left(1 -e^{-d\tau} \right)}{d \  \sigma^3 \left(1 - g e^{-d\tau} \right)}
\end{equation}

\begin{multline}
\label{equation:diff_int_sigma}
\frac{\partial}{\partial \sigma} \left(\frac{1 - e^{-d\tau} }{1 - g e^{-d\tau}} \right) = \\
\\
\frac{e^{-d\tau} \left[i u \rho \beta + 2 \hat{\alpha} \sigma \right] \left[ - \tau (\beta + d)^2 (1 - g e^{-d\tau}) + (1 - e^{-d\tau}) (2 \beta + \tau g (\beta + d)^2) \right] - 2 i u \rho d^2 e^{-d\tau} \left(1 - e^{-d\tau} \right)} {d  \ (\beta + d)^2 \left(1 - g e^{-d\tau} \right)^2}
\end{multline}

\begin{proof}
See Appendix \ref{section:demo_diff_wrt_sigma}.
\end{proof}

%%%%%%%%%%%%%%%%%%%%%%%%%%%%%%%

\section{Practical Implementation and Numerical Results}
\label{chap:implementation_numerical_results}

DML allows for fast training and accurate pricing. In the following, we start by constructing the Heston model in section~\ref{section:construct_heston}. Next, we describe the generator used to build the datasets utilised to train, validate, and test the neural networks in section~\ref{section:data_genration}. Section~\ref{section:data_handling} presents the data handling as well as the selection of the dataset. In section~\ref{section:traing_performance}, we compare the DML performance to the classical DL (without differentiation) one in the case of FFNN. We show that the DML outperforms the DL.
Furthermore, we introduce different regularisation techniques and apply them notably in the case of the DML in sections~\ref{section:traing_performance} and~\ref{section:regul_dml}. We compare their performance in reducing overfitting and improving the generalisation error. Finally, We perform a Hyperparameters tuning, select the best model and examine its performance in section~\ref{section:hype_tun}. Note that we provide a realistic and thorough evaluation of the obtained results in sections~\ref{section:data_handling},~\ref{section:traing_performance},~\ref{section:regul_dml} and~\ref{section:hype_tun}.

\subsection{Heston Pricer and Sensitivities} 
\label{section:construct_heston}

We start by constructing the Heston model. The semi-analytical formulas for the valuation of European options, in particular the normalised forward put price, as well as for its differentials with respect to inputs, involve complex integrals that can be calculated using several numerical methods by considering the truncation of the integral from a certain threshold. We use, in this work, the function \emph{quad\_vec} from the \emph{scipy.integrate} package in Python. It allows accurate computation of these integrals.\\

According to proposition \ref{proposition:bs_particular_case}, the Black-Scholes model is a particular case of the Heston one. So, to check our Heston prices implementation, we price European call and put options under both Black-Scholes and Heston models with the parameter values leading to Black-Scholes prices (i.e. satisfying (\ref{equation:bs_particular_case})) and obtain comparable prices. Note that we cannot substitute $\sigma = 0$ into the Heston pricing functions because that will lead to division by zero. We take small values of $\sigma$ to avoid this issue in the implementation.\\

We make a sanity check of our implementation of the semi-closed formulas for the different partial derivatives derived in section \ref{section:diff_p_prime_inputs}: We calculate them and find results similar to those obtained by performing finite differences. Recall that a finite difference is used as an approximation of the derivative in numerical differentiation. Indeed, the derivative of a function $f$ at a point $x$ is defined by:

\begin{equation*}
    {f'(x)=\lim _{h\to 0}{\frac {f(x+h)-f(x)}{h}}}
\end{equation*}

\subsection{Data Generation}
\label{section:data_genration}

This section describes the generator used to build the datasets utilised to train, validate, and test the neural networks.\\
 
Firstly, we generate synthetic dataset for the input parameters for the model ($\theta$, $\kappa$, $\sigma$, $\rho$ and $v_0$), the option and the market ones, which include $m$, $\tau$ and $r$ within appropriate ranges selected for each parameter. As mentioned in section~\ref{section: ANN_approach}, we use synthetic rather than historical data because it enables the neural networks to generate accurate results at future market states that might not have already happened.\\

In statistics and in machine learning, the i.i.d.~\footnote{independent and identically distributed} assumption about the features and labels $(x_i,y_i), i=1,\ldots,m$, is often made. Thus the data needs to be generated by a pseudo-random method such as simple Monte Carlo (MC), and the $m$ observations with respect to inputs need to be equally weighted when training the neural network. Moreover, the presence of clustering causing sparsity in the sampling is a source of important issues in the neural network training phase by not being able to cover all the value space. This becomes more significant when the input range is extensive, which is the case for time to maturities, with values from 0.05 to 20 years. Median Latin Hypercube (MLH) or Random Latin Hypercube (RLH) Sampling can prevent such issues. Median and Random Latin Hypercube methods rely on importance sampling to better spread the uniform random variables over the unity cube. Quasi-MC generators such as Sobol give low-discrepancy deterministic sequences that even better spread the samples. Still, the i.i.d property is lost, and special attention needs to be paid when applying it with stochastic gradient descent optimisation methods. In this work, we opt for the RLH approach to generate the parameters.

\subsubsubsection{Latin Hypercube Sampling}
\label{section:LHS}

Latin hypercube sampling (LHS) is a statistical method for generating random samples of parameter values with equal intervals.\\

The method performs the sampling by ensuring that each sample is positioned in a space $\Omega$ of dimension ${\displaystyle d}$ as the only sample in each hyperplane of dimension ${\displaystyle d-1}$ aligned on the coordinates which define its position. Each sample is therefore positioned according to the position of the samples positioned previously to ensure that they do not have common coordinates in the space $\Omega$. Thus, this method requires knowing the position of the samples already positioned.\\

The LHS method is entirely accepted because it presents optimisable non-collapsing space-filling properties and is flexible regarding data density and location. Hence, we decide to generate the parameters using the LHS. Table \ref{tab:sampling_range} summarises the sampling range of the model, option and market parameters.\\

\begin{table}[!htbp]
    \centering
   \begin{adjustbox}{width=0.8\textwidth}
    \begin{tabular}{|l|l|l|}
    \hline
    %  DNN   &  
     Parameters & Value Range & Generating Method\\
     \hline
    %   \multirow{8}{*}{DNN Input} &
      Log-moneyness, $m$ & [-2,2] & LHS\\
    %   &
       Time to maturity, $\tau$ & [0.05,20](years) & LHS\\
    %   & 
      Risk-free rate, $r$ & [-0.01,0.10] & LHS\\
    % &
    Mean reversion speed, $\kappa$ & [0.005,3] & LHS \\
    %   &
      Initial variance, $v_0$ & [0,1] & LHS \\
    %   & 
      Long-term mean-variance, $\theta$ & [0,1] & LHS\\
    %   & 
      Volatility of variance, $\sigma$ & [0.1,2] & LHS \\
    %  & 
     Correlation, $\rho$ & [-0.90,0] & LHS\\
        \hline
%   DNN Output & The normalised forward put price, $\hat{P}$ & - & Heston formula\\
%         \hline
       
    \end{tabular}
    \end{adjustbox}
    \caption{Sampling range for the parameters}
    \label{tab:sampling_range}
\end{table}

\newpage
Secondly, using the semi-closed pricing formulas derived in sections~\ref{section:nfpp} and~\ref{section:diff_p_prime_inputs}, we compute the corresponding option prices, which are considered as outputs, and the differential of prices with respect to inputs that are used as additional data. We generate 100~000 \ (100K) samples in two cases: in the first case, the Feller condition is satisfied, and in the second, it can be breached. \\ 

The data generation is computationally intensive. It takes 9 hours, 14 minutes and 51 seconds on a CPU server when the feller condition is satisfied and takes 10 hours, 34 minutes and 24 seconds when this condition can be breached. Nevertheless, this procedure is performed once. The learning stage on this synthetic data of the neural network takes a few minutes.\\

\subsection{Data Handling}
\label{section:data_handling}

\subsubsection{Data Splitting}
We randomly split the dataset into training, validation and test sets with a ratio of 80\% for training, 10\% for validation and 10\% for testing. \\

\subsubsection{Data Normalisation}

Normalisation is a technique applied as part of data preparation in DL. Data normalisation makes it possible to overcome the differences in the "norms" of the variables. Indeed, variables with large values can have a greater influence than variables with small values without being more significant. The normalisation aims to change data values to a joint scale without distorting differences or losing information. The most straightforward and most widely used technique treats each variable independently and calculates for each variable $x_j$ its mean value $\overline{x}_j$ and its standard deviation $\sigma_j$. Thus, the normalisation of the observation $x_i$ of this variable, identified by the element $x_{ij}$ of $X$, is normalised ($\tilde{x}_{ij}$) by the following expression:
\begin{equation*}
    \tilde{x}_{ij} = \frac{x_{ij} - \overline{x}_j}{\sigma_j}
\end{equation*}

The result of this normalisation on the set of observations $x_{ij}, i = 1, ..., m$ of the variable makes it possible to obtain a distribution of this variable having as properties a mean value of zero and a variance of one. In the literature, this transformation can be observed under the name of reduced centred transformation or even statistical normalisation.\\

We normalise the training data. When dealing with DL, the function performs normalisation for inputs and labels. In the case of DML, this function performs normalisation for inputs, labels and their derivatives. A reduced-centred transformation is applied to the inputs and the labels. Inspired by Huge and Savine~\cite{differential_ML}, differential labels are adjusted by multiplication of the standard deviations of inputs and division by the standard deviation of labels.

We apply the same transformation on the validation set, i.e. using the mean and standard deviation of the training dataset.

Because the model expects normalised inputs, we normalise the inputs of the test set by subtraction of the mean and division by the standard deviation of the inputs of the training set. When comparing predictions to the exact price, we scale back predictions to original units by multiplying them by the standard deviation and then adding the mean of the labels of the training set.\\

\subsubsection{Selection of the Dataset}

\subsubsubsection{Selection According to the Feller Condition}
The generalisation error when the Feller condition is satisfied is better than when this condition can be potentially breached, as we can see in Table~\ref{tab:feller_vs_vofeller}, where we provide the MSE on the test set between the exact and predicted prices when the Feller condition is satisfied and when it is breached for a FFNN with (respectively without) differentiation detailed in the next section. Note that all values are expressed in basis points (bp): 1bp=$10^{-4}$. 

\begin{table}[!htbp]
    \centering
    %\resizebox{12cm}{!}
    \begin{adjustbox}{width=0.8\textwidth}  
    \begin{tabular}{|c|c|c|}
    \hline
     \diagbox{Condition}{Network} & FFNN without differentiation & FFNN with differentiation\\
     \hline
     Feller condition satisfied &  2.60 & 0.44 \\
    \hline 
      Feller condition breached &  3.74 & 0.50\\
       
        \hline
    \end{tabular}
    \end{adjustbox}
    \caption{MSE (in bp) on the test set for a FFNN with (resp.without) differentiation when the Feller condition is satisfied and when it is breached.}
    \label{tab:feller_vs_vofeller}
\end{table}

This result confirms that the Heston model behaviour is sensitive to the Feller condition. The fact that the model parameters do not follow this condition may result in numerical instabilities while pricing options because volatility can reach null value causing the deterioration of the network performance.\\

In the sequel, we will work with the dataset generated when the feller condition is satisfied.

\subsubsubsection{Selection of the Dataset Size}

According to Table~\ref{tab:100K_vs_16k}, the generalisation error of the differential network trained on 16K examples is lower than when the network is trained on the whole dataset (i.e. 100K samples) since the MSE on the test set is 0.44 (respectively 2.66 ) bp with 16K (respectively 100K) samples. Unlike performing classical DL, the more samples in the dataset, the better the result. Moreover, The MSE on the differential network trained on 16K examples is 0.44 bp, lower than the MSE on the classical network with 6 times the training dataset size, which equals 1.73. This result illustrates the ability of differential networks to learn efficiently from small datasets showing one of their benefits.\\

Training with fewer samples reduces the computational time considerably and gives a good generalisation performance in the case of FFNN with differentiation; we then decide to retain a dataset of 16 384 samples. 

\begin{table}[!htbp]
    \centering
    \begin{adjustbox}{width=0.8\textwidth}  
    \begin{tabular}{|c|c|c|}
    \hline
    \diagbox{Dataset size}{Network} &  FFNN without differentiation & FFNN with differentiation\\
     \hline
     100 K & 1.73 & 2.66 \\
    \hline 
      16 K & 2.60 & 0.44 \\
        \hline
    \end{tabular}
    \end{adjustbox}
    \caption{MSE (in bp) on the test set for different dataset sizes in the case of a classical FFNN and a differential FFNN.}
    \label{tab:100K_vs_16k}
\end{table}

\subsection{Training Performance}
\label{section:traing_performance}

In this section, we compare the performance of the DML to the classical DL (without differentiation) one in the case of FFNN.\\

An FFNN is defined as a Sequential model in Pytorch. It is a plain stack of layers where each layer has exactly one input tensor and one output tensor.\\

To perform our comparison, three models are trained: The first is a basic sequential model using \emph{torch.nn.Sequential}, the second is an (explicit) FFNN (built from scratch) without differentiation, and the third is an (explicit) FFNN (built from scratch) with differentiation. All models are trained using the same neural network architecture and hyperparameters setting.
Note that we implement both the classical FFNN and the Differential FFNN from scratch using the PyTorch pythonic DL framework .

\subsubsubsection{Neural Network Architecture}
\label{section:NN_architecture}
In this section, we choose an FFNN architecture that consists of $4$ hidden layers with $50$ neurons (in each layer).\\

The training of the FFNN without differentiation is performed by optimising the weights and biases on the MSE cost function $\mathbf{C}$ defined by~(\ref{equation:cost_fct}), expressed in terms of the normalised labels $\tilde{Y}$ by:

\begin{equation*}
    \mathbf{C} = \frac{1}{m}\ \sum_{i=1}^m \left[\hat{y}_i (w) - \tilde{y}_i \right]^2
\end{equation*}

The training of the DML is performed by optimising the weights and biases on the MSE cost function $\mathbf{C}$ defined by~(\ref{equation:cost_fct_dnn}), expressed in terms of the normalised labels (respectively differential labels) $\tilde{Y}$ (respectively $\widetilde{\overline{X}}$) by:

\begin{equation*}
    \mathbf{C} = \frac{1}{m}\ \sum_{i=1}^m \left[\hat{y}_i (w) - \tilde{y}_i \right]^2 + \frac{1}{m} \ \sum_{i=1}^m \sum_{j=1}^n\ \frac{1}{||\widetilde{\overline{X}_j}||^2}\ \left[\widehat{{\overline{x}}}_{ij} (w) - \widetilde{\overline{x}}_{ij} \right]^2
\end{equation*}
Note that the optimal weights and biases are obtained by minimisation of $\mathbf{C}$ in $w$.\\

The "AdamW" optimiser is used. As suggested by Huge and Savine~\cite{differential_ML}, we choose an adaptative learning rate for the optimiser. Its profile as a function of the epoch is presented in Figure~\ref{fig:lr}. The number of epochs, batches per epoch and batch size are respectively equal to 50, 16 and 819.\\

\begin{figure}[!htbp]
    \centering
    \includegraphics[width=11cm, height=5cm]{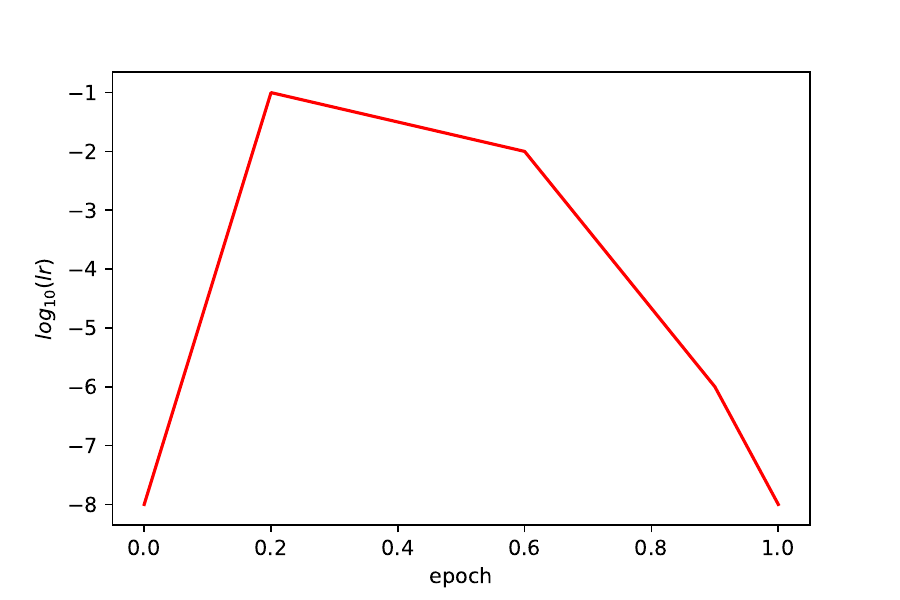}
    \caption{Base 10 logarithm of the learning rate.}
    \label{fig:lr}
\end{figure}

Each hidden layer in the FFNN without differentiation and the first half of the twin network as part of the FFNN with differentiation uses \emph{Softplus} activation functions with the final output layer a simple linear function to give a real-valued output. We choose a \emph{Softplus} activation function because it is a smooth approximation to the \emph{ReLU} activation function and allows the network to return positive values as output. 

The activation function for the second half of the twin network is the \emph{Sigmoid} (the derivative of the  \emph{Softplus} function) as part of the FFNN with differentiation. \\

The training is performed by initialising weights (respectively biases) with \emph{Kaiming} (respectively \emph{uniform}) initialisation.\\

The neural network architecture and hyperparameters are summarised in Table~\ref{tab:nn_architecture1}.

\begin{table}[!htbp]
    \centering
    \begin{adjustbox}{width=0.6\textwidth}  
    \begin{tabular}{|l|l|}
    \hline
    Parameter   &  Option\\
    \hline 
      Hidden layers & 4\\
       Neurons per layer & 50\\
     Epoch & 50 \\
     Batches per epoch & 16\\
    Batch size &  819 \\
    Learning rate & Described above\\
    Activation function for the forward pass & Softplus\\
    Activation function for the backward pass & Sigmoid\\
    Weights initialisation & Kaiming uniform \\
    Biases initialisation & Uniform \\
   Optimiser & AdamW \\
    Loss function & MSE\\
        \hline
    \end{tabular}
    \end{adjustbox}
    \caption{NNs architecture and hyperparameter setting}
    \label{tab:nn_architecture1}
\end{table}

\subsubsection{Accuracy of the Different Networks}

In order to assess the performance of the networks, firstly, we look at the evolution of training and validation losses according to the epoch. Figure~\ref{fig:loss_comparing_3models} depicts the results of the performance for training and validation. It can be seen that as the number of epochs increases, both training and validation losses decrease showing that there is a good convergence in the 3 models. \\

\begin{figure}[!htbp]
    \centering
    \includegraphics[width=15cm, height=10cm]{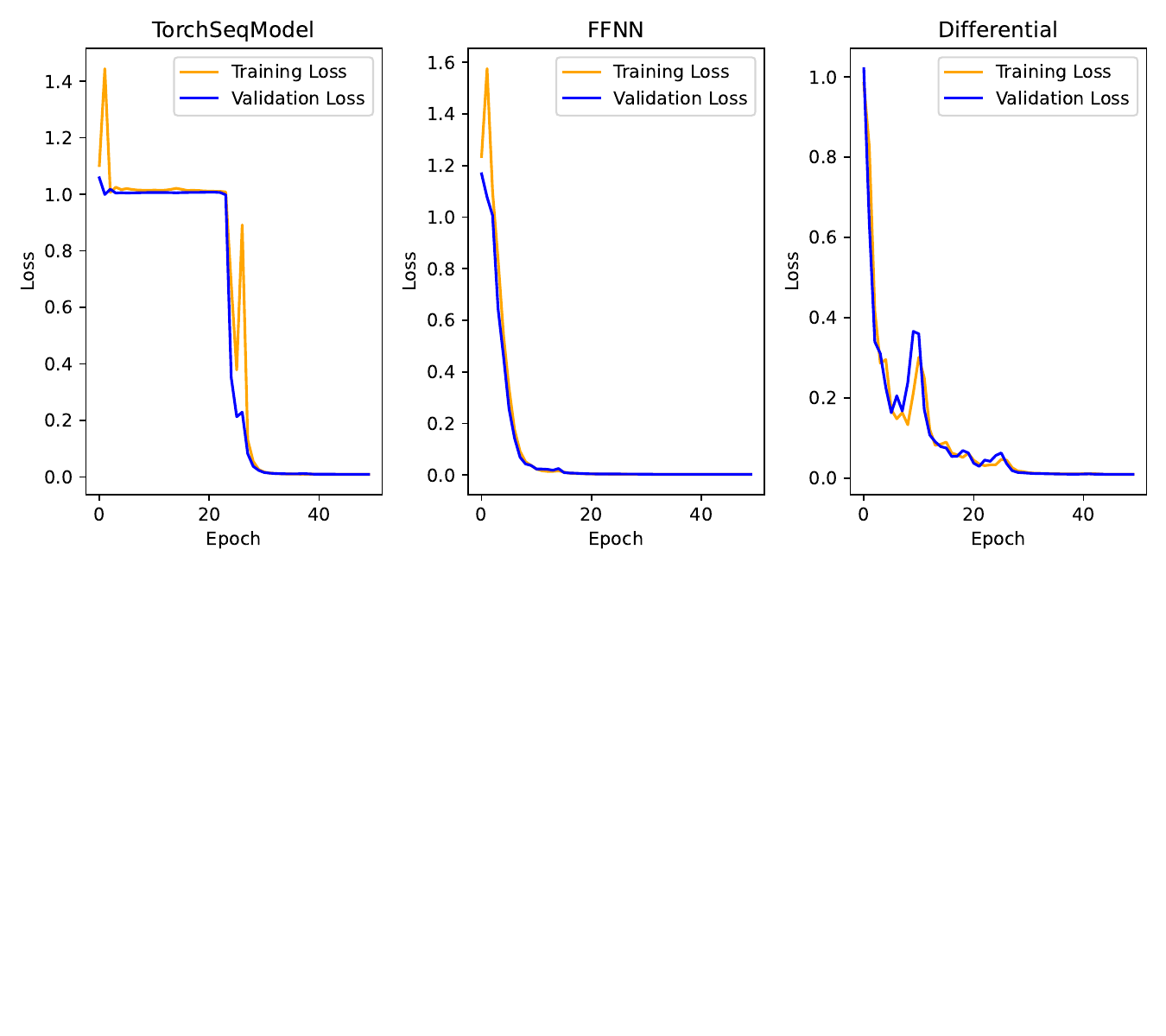}
    \caption{Training and validation losses according to the epoch for three NNs. Left: basic sequential NN, middle: classical FFNN, right: FFNN with differentiation.}
    \label{fig:loss_comparing_3models}
\end{figure}

Secondly, we test each model with the test dataset, in particular, we compare predicted prices with the exact ones. Results are displayed in Table~\ref{tab:mse_comp3models} where we display the MSE on the test set for each model. The lower the MSE, the better performance by the model. It can be seen that the FFNN model without differentiation that we implement outperforms the basic sequential model; this may be due to the fact that the initialisation of weights and biases method that we implement is better than that of the basic sequential model. Furthermore, it is clear that the DNN gives the best result and outperforms the classical NN by far since the MSE is much lower. The latter result can be explained by the presence of the penalty allocated for wrong differentials in the cost function that acts as a regulariser, as mentioned in section~\ref{section:train_twin_network} showing that differentials act as an efficient regulariser.\\

 \begin{table}[!htbp]
    \centering
    \begin{adjustbox}{width=0.4\textwidth}  
    \begin{tabular}{|l|l|l|}
    \hline
    Neural Network   &  MSE (bp)\\
    \hline 
      Basic Sequential NN & 6.29\\
       Classical FFNN & 2.60\\
     DNN & 0.44 \\
        \hline
    \end{tabular}
    \end{adjustbox}
    \caption{NNs performance on the test set}
    \label{tab:mse_comp3models}
\end{table}

\subsubsection{Regularisation Techniques}

Overfitting is a common problem when training a DL model. 
To counter this problem,  a set of methods that optimise the learning of a DL model called regularisation has been developed over the years.

\subsubsubsection{Regularisation L1 and L2}

The first frequently used techniques are L1 (Lasso) and L2 (Ridge) regularisation. These methods consist in forcing the network to have weights of small sizes. They are based on the idea that if an ideal network exists to solve the problem, it is possible to modify all the weights of this network by the same constant without changing the final result. Among this infinity of possible networks, it is preferable to choose the network with the weights closest to zero because it is the simplest.\\

To apply the L1 (respectively L2) regularisation, simply add a term to the cost function: $\mathbf{C}_1 = \mathbf{C} +  \lambda |\mathbf{W}|$ \ (respectively \ $\mathbf{C}_2 = \mathbf{C} + \lambda \mathbf{W}^2$), 
where $\mathbf{C}_1$ and $\mathbf{C}_2$ are the cost function with L1 and L2 regularisation, respectively, $\mathbf{C}$ is the initial cost function, $\mathbf{W}$ represents the network weights, and $\lambda$ is a hyperparameter to control the strength of the regularisation. \\

By limiting the size of the weights, this regularisation decreases the expressive capacity of the network. Indirectly, this also reduces the possibilities for the network to memorise the training data.\\

In this work, we implement L1 and L2 regularisations but do not activate them because we use AdamW as an optimiser.

\subsubsubsection{AdamW}

AdamW~\cite{adamw} is a variation of Adam optimiser. Both optimisers add a penalty L2 to regularise the optimised neural network weights. The only difference is that Adam defines the L2 regularisation through the cost function, whereas AdamW implements it directly into the weight update rule.\\

Another technique that we apply to the above models is called Dropout.

\subsubsubsection{Dropout}
\textbf{\emph{Technique Presentation}}\\

Dropout~\cite{dropout} is a technique to reduce overfitting when training the model.
The term \emph{dropout} refers to the removal of neurons in the layers of a DL model.
In fact, specific neurons are temporarily deactivated in the network, as well as all its incoming and outgoing connections, as illustrated in ﬁgure~\ref{fig:dropout_architecture}.

\begin{figure}[!htbp]
    \centering
    \includegraphics[width=12cm, height=5cm]{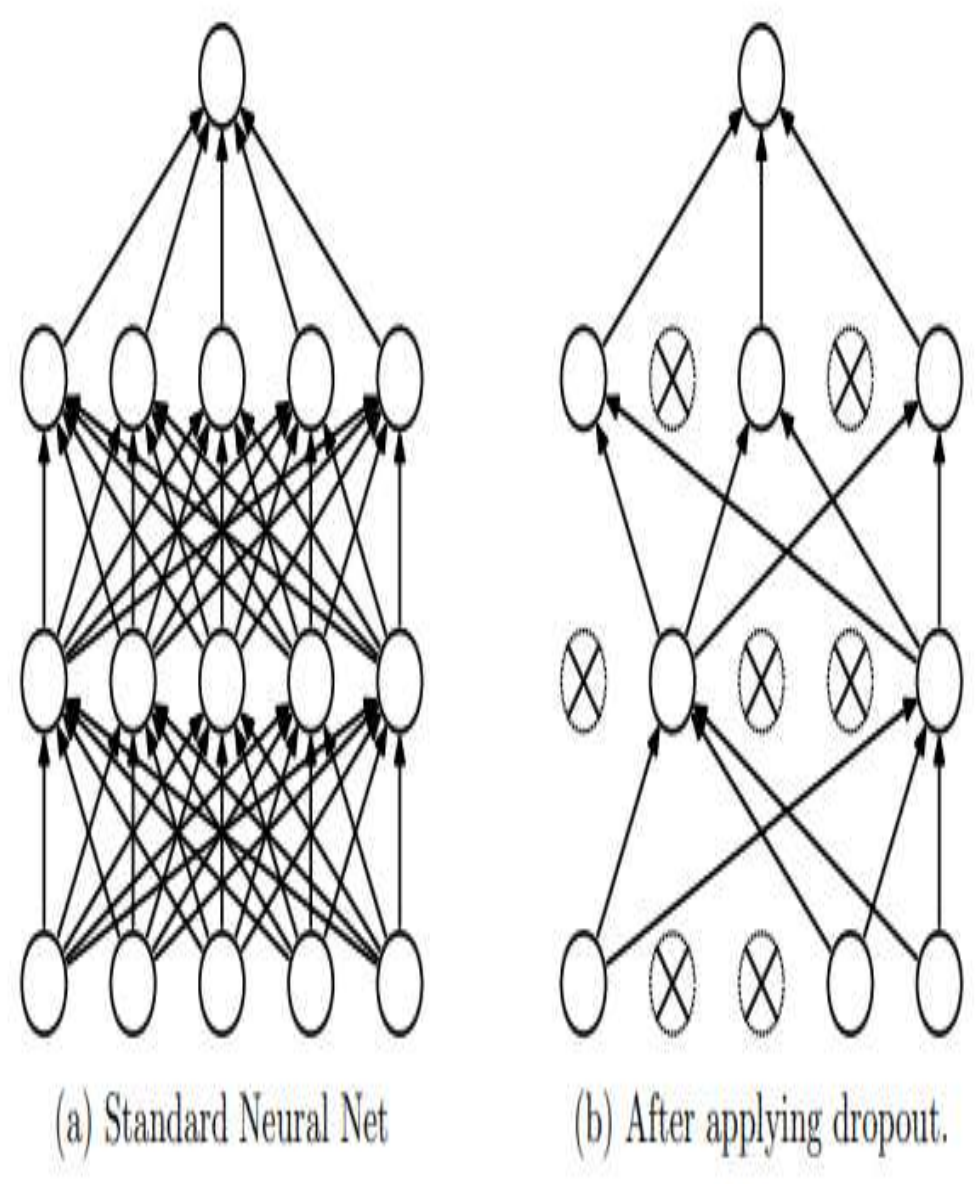}
    \caption{Dropout architecture  as presented in~\cite{dropout}. Left: a neural network with 2 hidden layers, right: the network after applying dropout to the network on the left.}
    \label{fig:dropout_architecture}
\end{figure}

The choice of neurons to be deactivated is random. A probability $p$ is assigned to all the neurons, which determines their activation. When $p = 0.1$, each neuron has a 1 in 10 chance of being deactivated. At each epoch,  this random deactivation is applied. At each pass (forward propagation), the model will learn with a configuration of different neurons; the neurons activate and deactivate randomly. This procedure effectively generates slightly different models with different neuron configurations at each iteration. The idea is to disturb the characteristics learned by the model. Usually, the model's training relies on the synchronicity of the neurons with the dropout; the model must exploit each neuron individually, its neighbours being able to be deactivated randomly at any time. However, the dropout is active only during the training of the model. During the tests, each neuron remains active, and its weight is multiplied by the probability $p$.\\

When implementing the dropout regularisation method to perform DL or DML (via a from scratch implementation), one has (only when the training mode is active) to generate a mask vector according to a Bernoulli distribution (with p = dropout rate) and then normalise (i.e. divide) it by $1-p$ for each layer (impacted by the dropout technique) in the forward pass. The backward pass providing the differential of labels will use the masks generated in the forward mode.\\

\textbf{\emph{Numerical Results}}\\

Figure~\ref{fig:loss_dropout} displays the evolution of training and validation losses before and after applying dropout for each model. We note that there is a better convergence of both training and validation losses after applying dropout in the case of the basic sequential NN because both losses decrease as the number of epochs increases while both losses fluctuate a little before applying dropout indicating that dropout reduces overfitting. Despite both training and validation losses decrease after applying dropout in the case of classical FFNN and differential FFNN, they are more stable before applying dropout showing that the dropout enlarges overfitting.\\

Table~\ref{tab:mse_dropout} showsthe MSE on the test set for each model before and after applying dropout. We observe that applying dropout to the basic sequential model improves the generalisation error, but this error remains higher than that of both FFNN models with and without differentiation before applying dropout. In addition, we notice that dropout deteriorates the performance of both FFNN models with and without differentiation since the MSE after applying dropout is higher than before applying it.\\ 

\begin{figure}[!htbp]
    \centering
    \includegraphics[width=15cm, height=10cm]{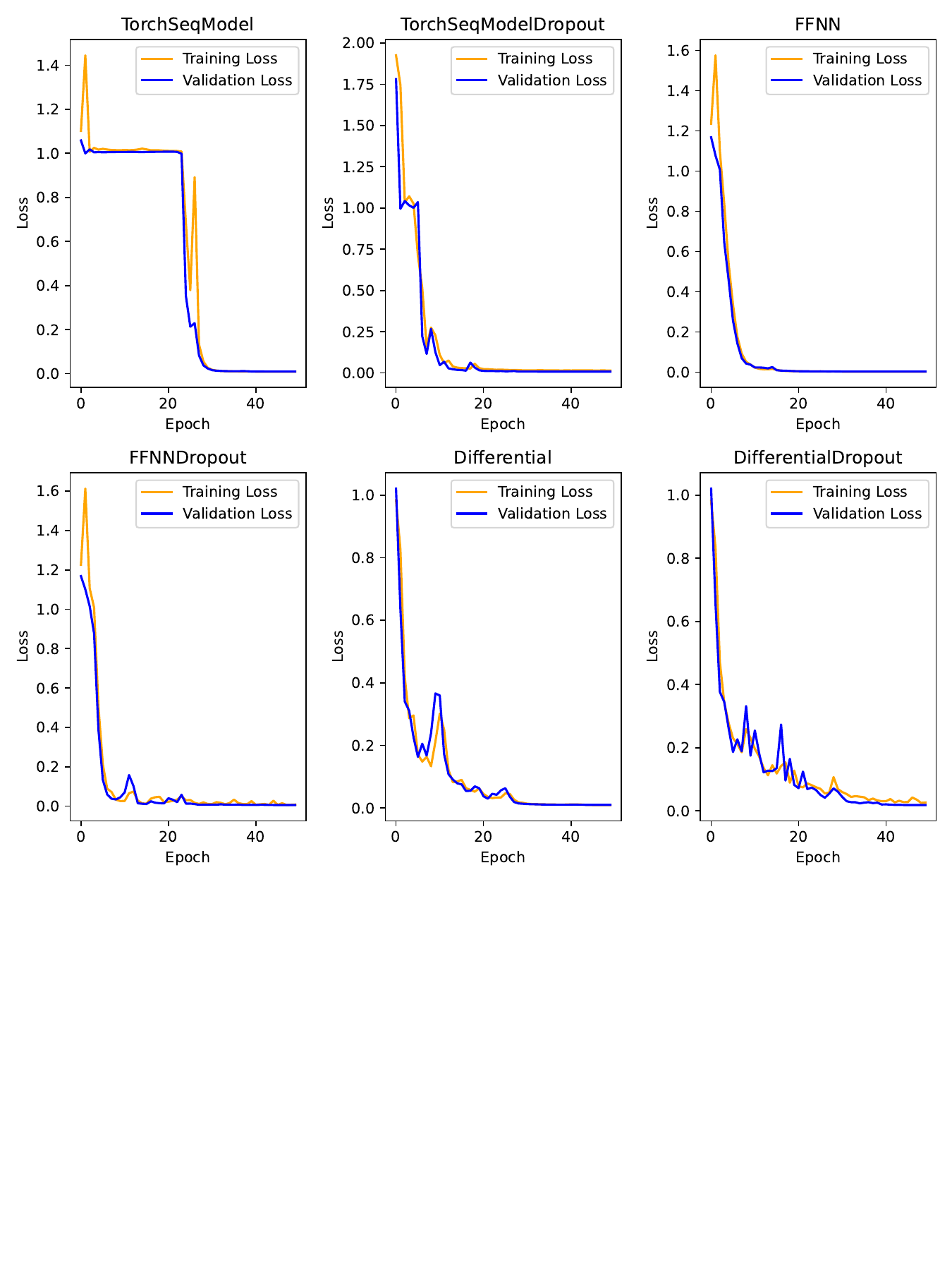}
    \caption{Training and validation losses according to the epoch for different NNs. Top left: basic sequential NN, top middle: basic sequential after applying dropout, top right: classical FFNN, bottom left: classical FFNN after applying dropout, bottom middle: FFNN with differentiation, bottom right:  FFNN with differentiation after applying dropout.}
    \label{fig:loss_dropout}
\end{figure}

 \begin{table}[!htbp]
    \centering
    \begin{adjustbox}{width=0.5\textwidth}  
    \begin{tabular}{|l|l|l|l|l|l|}
    \hline
    Neural Network   &  MSE (bp)\\
    \hline 
      Basic Sequential NN & 6.29\\
    basic sequential NN after applying dropout & 3.89 \\
    FFNN & 2.60 \\
    FFNN after applying dropout & 4.82 \\
     DNN & 0.44 \\
     DNN after applying dropout & 0.72 \\
        \hline
    \end{tabular}
    \end{adjustbox}
    \caption{NNs performance on the test set}
    \label{tab:mse_dropout}
\end{table} 

The results of this section indicate that DML outperforms traditional DL vastly and that dropout deteriorates the performance of the network. In the sequel, we will focus on DML.

\subsection{Regularisation Techniques Applied to DML}
\label{section:regul_dml}

In order to control the complexity of the differential network and enhance the generalisation performance and the learning process speed, we apply different regularisation techniques and compare their performances in reducing overfitting and improving the generalisation error. 

\subsubsection{Differential Wide \& Deep Learning}
\label{section:diff_wide_deep}
\subsubsubsection{Technique Presentation}

We train the network with a particular architecture introduced by Google as part of the recommender systems called "Wide \& Deep"~\cite{wide_and_deep}. This concept combines the strengths of memorisation and generalisation. It jointly trains wide linear models for memorisation alongside deep neural networks for generalisation. \\

The wide component accepts a set of features and feature transformations to remember feature interactions. Moreover, with less feature engineering, the deep component is generalised to invisible feature combinations. The wide part and the deep part are concatenated by summing their eventual log odds as a prediction, which eventually comes together in a single running loss function for the joint formation. \\

In the context of DML, as illustrated in figure~\ref{fig:wideanddeep_architecture}, Huge and Savine~\cite{differential_ML} showed that the output layer is a linear regression on the combination of the last hidden layer of the deep neural network $z_{L}$ and the wide layer, which is a set of fixed functions $g$:

\begin{equation*}
\hat{h} (x; \mathbf{W}) = z_{L} (x; \mathbf{W}_{hidden}) \mathbf{W}_{deep} + g(x) \mathbf{W}_{wide} 
\end{equation*}
Where $\hat{h}$ is an approximation of the target function $h$, $x$ are the inputs, and $\mathbf{W}$ are the learnable weights. Thus the differentials of predictions with respect to inputs are:

\begin{equation*}
\frac{\partial \hat{f}(x;\mathbf{W})}{\partial x_j} = \frac{\partial z_{L}(x;\mathbf{W}_{hidden})}{\partial x_j} \mathbf{W}_{deep} + \frac{\partial g(x)}{x_j} \mathbf{W}_{wide}
\end{equation*}

where $\frac{\partial z_{L}}{\partial x}$ are the differentials of the deep network and $\frac{\partial g}{\partial x}$ are the derivatives of $g$. Therefore, the wide \& deep architecture can be implemented in DML as follows: 

\begin{itemize}
    \item Adding a regression term $g(x) \mathbf{W}_{wide}$ to the output of the deep neural network.
    \item Adjusting the gradients of output with respect to inputs that are computed by backpropagation through the deep neural network by the term $\frac{\partial g(x)}{\partial x_j} \mathbf{W}_{wide}$. 
\end{itemize}

\begin{figure}[!htbp]
    \centering
    \includegraphics[width=12cm, height=12cm]{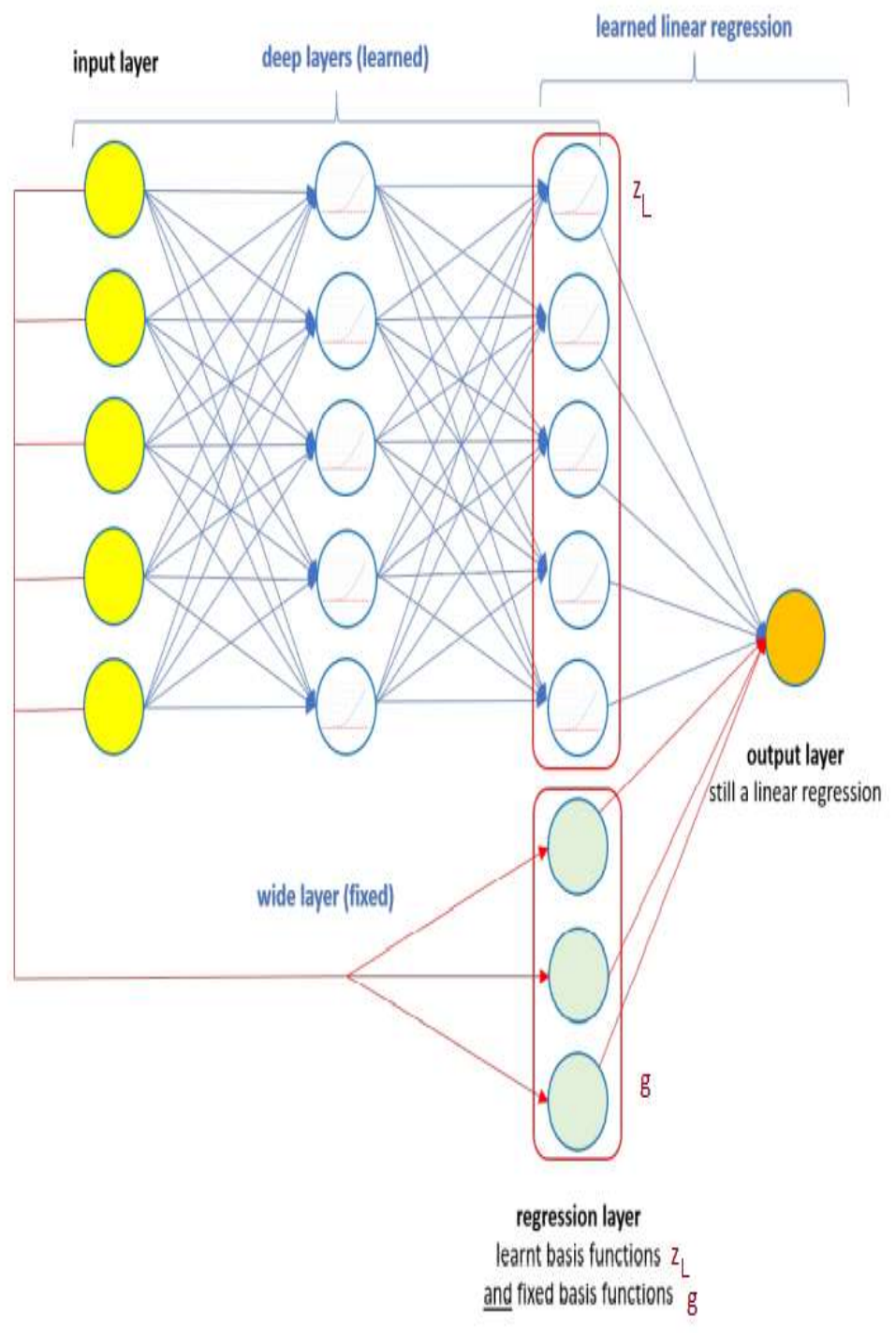}
    \caption{Wide \& deep architecture; as presented in~\cite{differential_ML}}
    \label{fig:wideanddeep_architecture}
\end{figure}

\subsubsubsection{Numerical Results}

Figure~\ref{fig:loss_wideanddeep} represents the evolution of training and validation losses according to the epoch for both differential deep and differential wide \& deep networks. We can see that, in the case of the differential wide \& deep network, overfitting does not occur. In addition, both losses converge faster towards zero and are more stable than in the case of the differential deep network meaning that the wide \& deep architecture reduces overfitting and improves the training of the network.\\

Table~\ref{tab:mse_wideanddeep} illustrates the MSE on the test set for both differential deep and differential wide \& deep networks. We notice that the wide \& deep architecture improves the generalisation error.

\begin{figure}[!htbp]
    \centering
    \includegraphics[width=14cm, height=8cm]{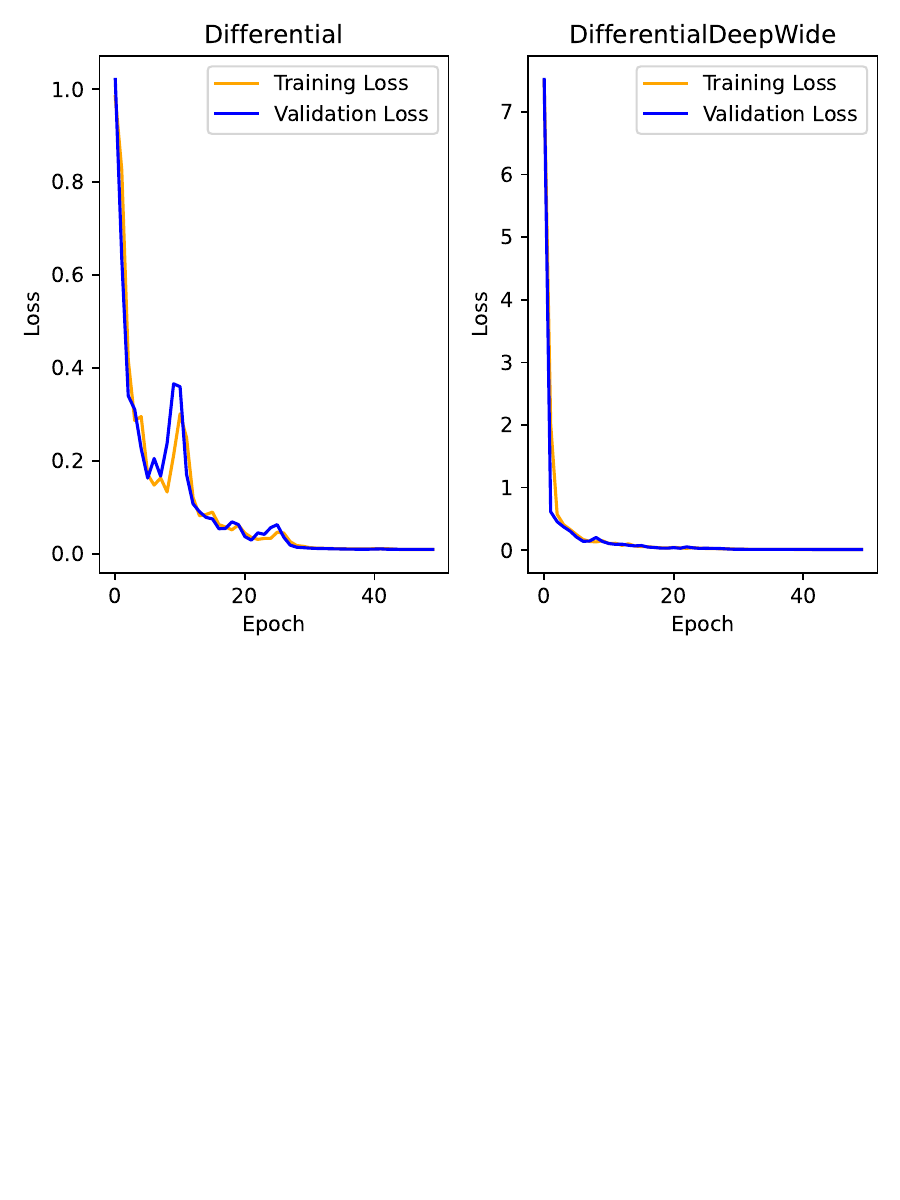}
    \caption{Training and validation losses according to the epoch. left: Differential network, right: Differential wide \& deep network.}
    \label{fig:loss_wideanddeep}
\end{figure}

\begin{table}[!htbp]
    \centering
    \begin{adjustbox}{width=0.5\textwidth}  
    \begin{tabular}{|l|l|}
    \hline
    Neural Network   &  MSE (bp)\\
    \hline       
     DNN & 0.44 \\
     Differential wide \& deep network & 0.38 \\
        \hline
    \end{tabular}
    \end{adjustbox}
    \caption{DNN performance on the test set}
    \label{tab:mse_wideanddeep}
\end{table} 

% \begin{figure}[!htbp]
%     \centering
%     \includegraphics[width=14cm, height=8cm]{./figures/mse_testing_set_widedeep.pdf}
%     \caption{DNN performance on the test set. Left: Differential network, right: Differential wide \& deep network.}
%     \label{fig:mse_wideanddeep}
% \end{figure}

\subsubsection{Gradient Clipping}
\label{section:exploding_grad}

We notice a problem with exploding gradients. This problem refers to a significant increase in the norm of the gradient during training that goes, for instance, from 0.6988 to 314.8369 and even reaches 1107.8723 when the seed is 3283. Long-term constituents may grow exponentially more than short-term ones and cause an explosion that produces such phenomenons. This generates an unstable network that cannot learn from the training data, making the gradient descent step difficult to execute. Poor performance of the DML is achieved in this case. Furthermore, we notice that the performance deteriorates as the norm of the gradient increases. Indeed, the MSE on the test set is equal to 674.53 bp when the gradient norm reaches 1107.8723 for a seed equals 3283, whereas the MSE is less significant (equals to 4.81) when the gradient norm reaches its highest at 17.7358 for a seed equals to 17653.\\

To prevent the problem of exploding gradients in deep networks, a popular technique is to clip them during backpropagation  so that they never exceed a certain threshold.  It is called gradient clipping.  We apply this technique to our DNN.\\ 

In Figure~\ref{fig:loss_gradclip}, we represent the evolution of training and validation losses according to the epoch for the DNN before and after applying gradient clipping and the differential wide \& deep network. We can see that, in the case of the differential wide \& deep network, both losses converge faster towards zero and are more stable than in the case of the differential deep network before and after applying gradient clipping.  Nevertheless, this does not imply that the differential wide \& deep network provides a better generalisation performance than the differential deep network after applying gradient clipping.\\

Indeed, Table~\ref{tab:mse_gradclip}, where we illustrate the MSE on the test set for DNN before and after applying gradient clipping and differential wide \& deep network, demonstrates that gradient clipping technique improves much better the generalisation error than the wide \& deep architecture, the MSE of the DNN is even divided by 2.

\begin{figure}[!htbp]
    \centering
    \includegraphics[width=14cm, height=7cm]{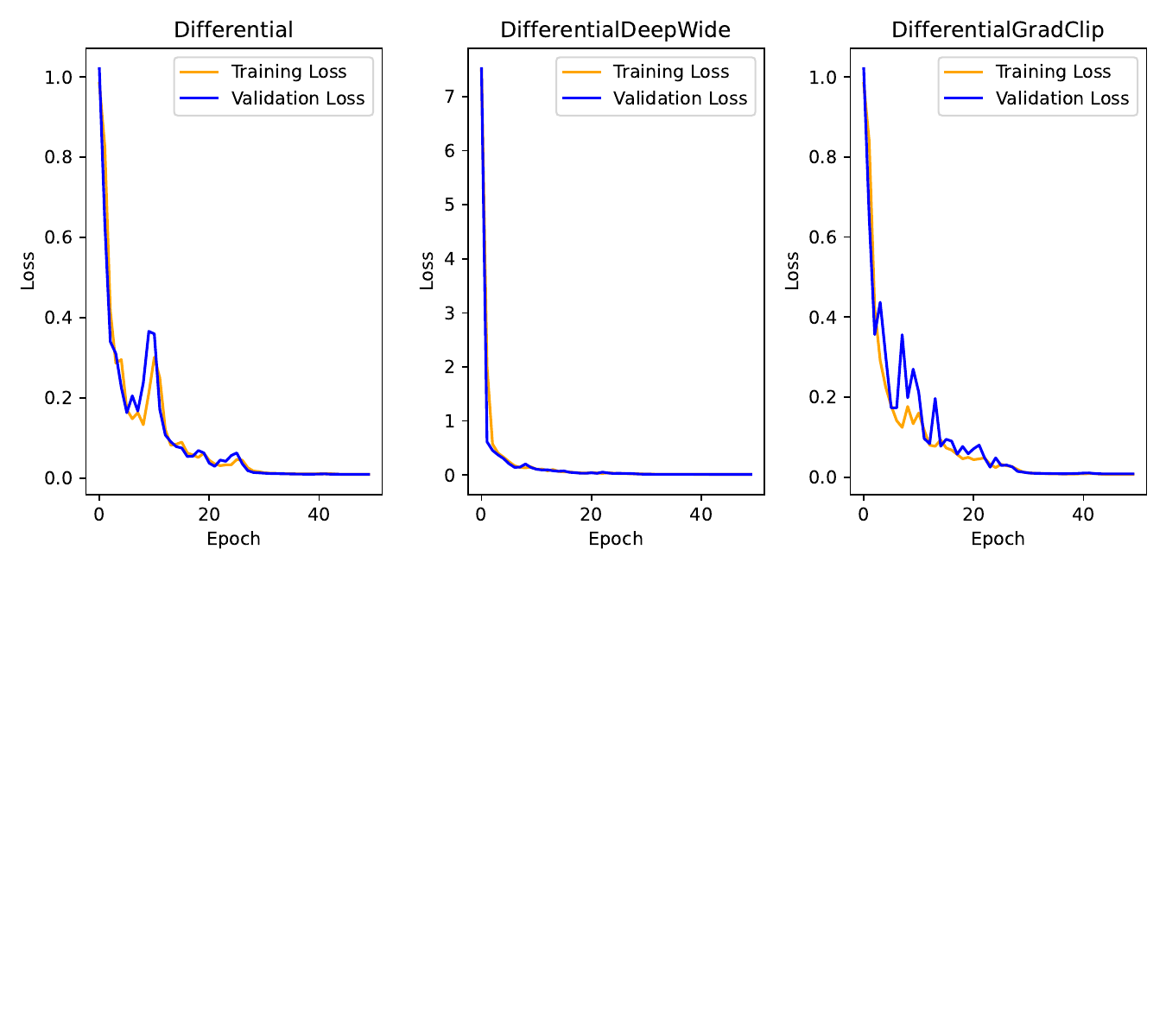}
    \caption{DNN training and validation losses according to the epoch. Left: DNN, middle: Differential wide \& deep network, right: DNN after applying gradient clipping.}
    \label{fig:loss_gradclip}
\end{figure}

\begin{table}[!htbp]
    \centering
    \begin{adjustbox}{width=0.5\textwidth}  
    \begin{tabular}{|l|l|l|}
    \hline
    Neural Network   &  MSE (bp)\\
    \hline       
     DNN & 0.44 \\
     Differential wide \& deep network & 0.38 \\
     DNN after applying gradient clipping & 0.20 \\
     
        \hline
    \end{tabular}
    \end{adjustbox}
    \caption{DNN performance on the test set}
    \label{tab:mse_gradclip}
\end{table} 

% \begin{figure}[!htbp]
%     \centering
%     \includegraphics[width=14cm, height=7cm]{./figures/mse_testing_set_gradclip.pdf}
%     \caption{DNN performance on the test set. Left: DNN, middle: Differential wide \& deep network, right: DNN after applying gradient clipping.}
%     \label{fig:mse_gradclip}
% \end{figure}

\subsubsection{Early Stopping}
\subsubsubsection{Technique Presentation}

Another method of regularisation widely used is \emph{early stopping}. Indeed, unlike the previous techniques, this technique does not modify or improve the network topology but stops the network's training when the model starts to overfit the training data. It is based on the periodic evaluation of the generalisation performance of the neural network during its learning phase and observations not used during learning. Thus during learning, validation observations make it possible to observe the generalisation error simultaneously with the learning error. Thus, if the loss of the validation set begins to increase while the loss of the training set continues to decrease, we can then consider that the network is entering the overfitting zone. In order to avoid it, it suffices to stop learning at the moment when the two losses diverge.

\subsubsubsection{Numerical Results}

In Figure~\ref{fig:loss_earlystop}, we represent the evolution of training and validation losses according to the epoch for DNN before and after applying early stopping, DNN after applying gradient clipping and after applying both gradient clipping and early stopping and differential wide \& deep network before and after applying early stopping. In Table~\ref{tab:mse_earlystop}, we illustrate the MSE on the test set for DNN before and after applying early stopping, DNN after applying gradient clipping and after applying both gradient clipping and early stopping and differential wide \& deep network before and after applying early stopping.\\

According to these figures, the behaviour of the losses of the different networks is the same after applying early stopping. In addition, the early stopping technique does not improve the generalisation error. This behaviour is expected in the case of the differential wide \& deep network because there is no overfitting, as shown in section~\ref{section:diff_wide_deep}. Regarding DNN and DNN with gradient clipping, this may be due to the choice of the patience value (i.e. the epochs number without amelioration, after which training will be early stopped.), which equals 10. We may expect an improvement in the performance of the networks if we take a smaller patience value. We will explore this issue in the next section by performing a grid search on a search space including different patience values.

\begin{figure}[!htbp]
    \centering
    \includegraphics[width=17cm, height=10cm]{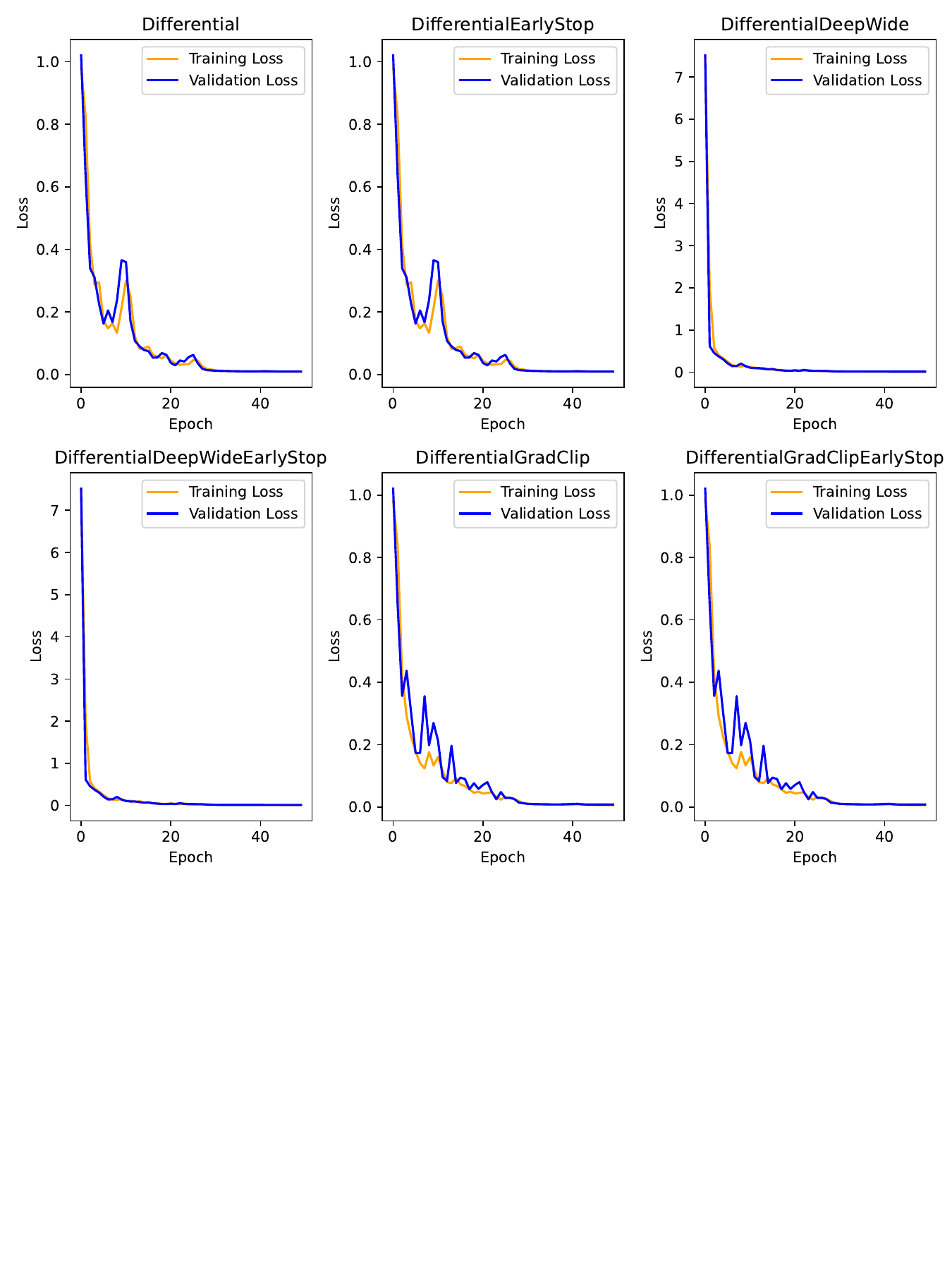}
    \caption{DNN training and validation losses according to the epoch. Top left: DNN, top middle: DNN after applying early stopping, top right: Differential wide \& deep network, bottom left: Differential wide \& deep network after applying early stopping, bottom middle: DNN after applying gradient clipping, bottom right: DNN after applying gradient clipping and early stopping.}
    \label{fig:loss_earlystop}
\end{figure}

\begin{table}[!htbp]
    \centering
    \begin{adjustbox}{width=0.5\textwidth}  
    \begin{tabular}{|l|l|l|l|l|l|}
    \hline
    Neural Network  &  MSE (bp)\\
    \hline       
     DNN & 0.44 \\
     DNN after applying early stopping & 0.45 \\
     Differential wide \& deep network & 0.38 \\
     Differential wide \& deep network after applying early stopping & 0.38 \\
     DNN after applying gradient clipping & 0.20 \\
     DNN after applying gradient clipping and early stopping & 0.20 \\   
        \hline
    \end{tabular}
    \end{adjustbox}
    \caption{DNN performance on the test set}
    \label{tab:mse_earlystop}
\end{table} 

% \begin{figure}[!htbp]
%     \centering
%     \includegraphics[width=17cm, height=8cm]{./figures/mse_testing_set_earlystop.pdf}
%     \caption{DNN performance on the test set. Top left: DNN, top middle: DNN after applying early stopping, top right: Differential wide \& deep network, bottom left: Differential wide \& deep network after applying early stopping, bottom middle: DNN after applying gradient clipping, bottom right: DNN after applying gradient clipping and early stopping.}
%     \label{fig:mse_earlystop}
% \end{figure}

\subsection{Hyperparameters Tuning}
\label{section:hype_tun}

Hyperparameters are tunable parameters that drive the learning process. Model performance is highly dependent on them. Hyperparameters adjustment (tuning) is a crucial step in machine learning and requires statistical work to determine the hyperparameters configuration that produces the best performance. The process can be treated as a search problem. Grid Search is one of the standard approaches we use in this work. The principle consists of a list of values established for each hyperparameter. The model is trained for each combination of hyperparameters, and the cost function on the validation dataset is calculated each time. The combination that minimises this function is chosen. 

\subsubsection{Search Space}

DNN has a large number of hyperparameters, and executing a grid search on all of them is very computationally expensive. We then choose to perform a grid search on two search spaces. The first one is displayed in Table~\ref{tab:hyper_space}.\\

\begin{table}[!htbp]
    \centering
    \begin{adjustbox}{width=0.8\textwidth}  
    \begin{tabular}{|l|l|}
    \hline
Hyperparameter & Value\\
\hline 
\hline
 \multirow{4}*{Regularisation procedure} & wide \& deep \\ 
       & wide \& deep with early stopping  \\
      & gradient clipping \\
      & gradient clipping and early stopping \\
\hline
Number of hidden layers & $\in \{4, 6, 8\}$ \\
\hline
Number of Neurons per layer & $\in \{30, 50, 100\}$\\
\hline
Number of epochs & $\in \{30, 50, 100\}$\\
\hline
Early stopping patience & $\in \{5, 10, 15, 20\}$\\
\hline
Gradient clipping threshold & $\in \{0.5, 1, 2, 4\}$\\
\hline
Learning rate schedule & $\in \{0.1, 0.01, 10^{-3}, 10^{-4}, 2 \times 10^{-5}, 5 \times 10^{-5} , 10^{-6}, 10^{-8}\}$\\\hline
    \end{tabular}
    \end{adjustbox}
    \caption{Hyperparameters search space}
    \label{tab:hyper_space}
\end{table}

The second search space is the same as the first one. The only difference is that the learning rate schedule is no longer constant but is defined as in section~\ref{section:NN_architecture}.\\

Note that a model checkpoint callback has been introduced in the early stopping strategy, where the best weights and biases with respect to the validation loss are saved and re-used. 

\subsubsection{Best Model}

After almost four days on a GPU server, we record that the early stopping regularisation technique does not improve the network performance, and we obtain the following results: the best MSE on the validation set after performing the first hyperparameter tuning equals 0.0058, whereas the best one obtained after the second hyperparameter tuning equals 0.0076.\\

The hyperparameters configuration that reaches the best score is shown in Table~\ref{tab:best_model}.

\begin{table}[!htbp]
    \centering
    \begin{adjustbox}{width=0.6\textwidth}  
    \begin{tabular}{|l|l|}
    \hline
Hyperparameter & Value\\
\hline 
\hline
 Regularisation procedure & gradient clipping \\
\hline
Number of hidden layers & 6 \\
\hline
Number of Neurons per layer & 50\\
\hline
Number of epochs & 97\\
\hline
Early stopping patience & None\\
\hline
Gradient clipping threshold & 4\\
\hline
Learning rate schedule & 0.01 \\\hline
    \end{tabular}
    \end{adjustbox}
    \caption{The configuration of the best model}
    \label{tab:best_model}
\end{table}

As illustrated in Figure~\ref{fig:loss_bestmodel}, after 97 epochs, both training and validation losses converge without overfitting. The results and performance of the best-trained DNN are summarised in Table~\ref{tab:DML_performance}.

\begin{figure}[!htbp]
    \centering
    \includegraphics[width=14cm, height=8cm]{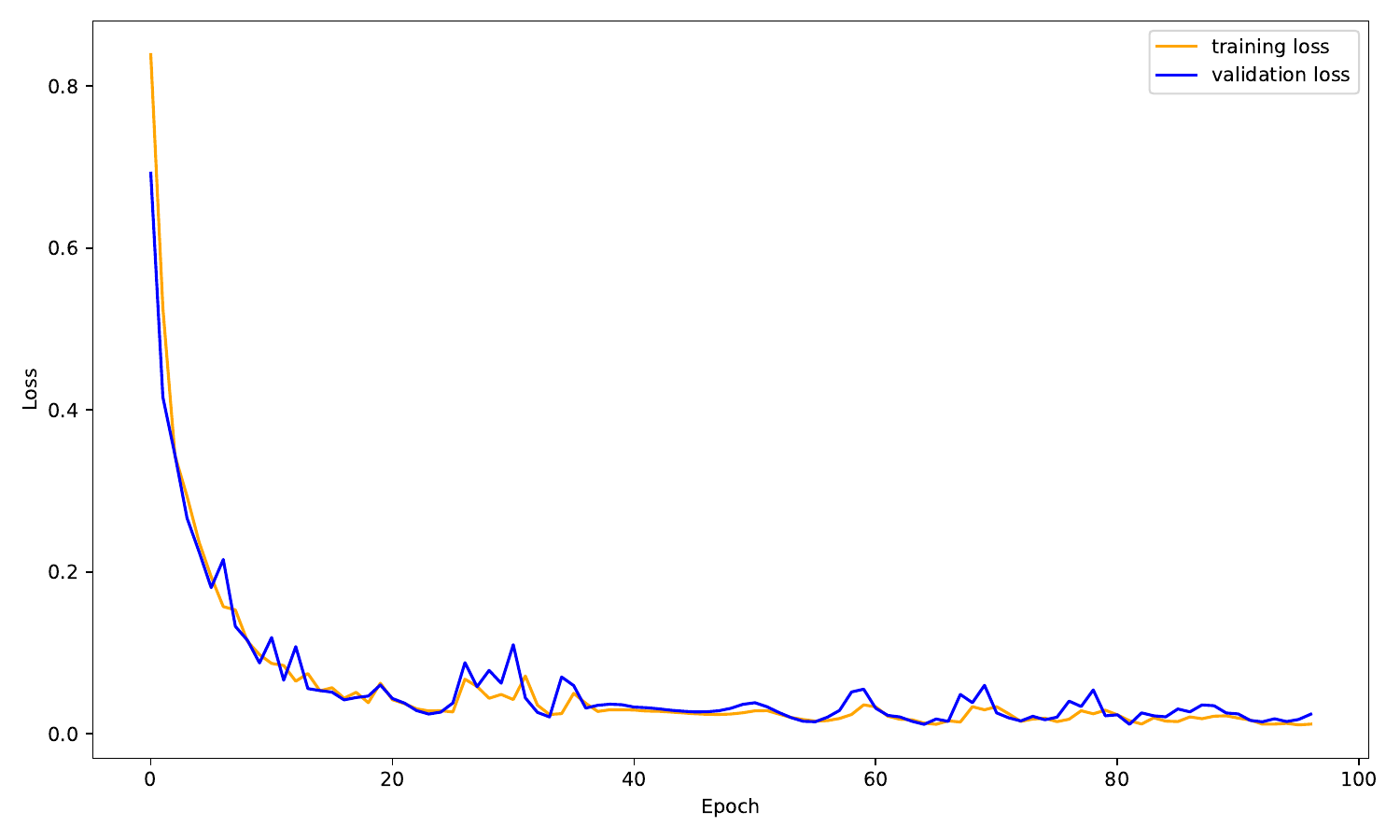}
    \caption{DNN training and validation losses according to the epoch of the best model.}
    \label{fig:loss_bestmodel}
\end{figure}

\begin{table}[!htbp]
    \centering
    \begin{adjustbox}{width=0.4 \textwidth}
    \begin{tabular}{|l|l|}
    \hline
DNN & MSE\\
\hline
 Training & 0.0061 \\
\hline
 Validation & 0.0058 \\
\hline
Testing & $3.72 \times 10^{-5}$ \\
\hline
    \end{tabular}
    \end{adjustbox}
    \caption{DNN performance of the best model}
    \label{tab:DML_performance}
\end{table}

It can be seen from Table~\ref{tab:DML_performance} that the performance of the best-trained DNN is robust and extremely high performing. Indeed, the performance on the training set is similar to the one on the validation set in terms of MSE. Furthermore, the MSE on the test set equals $3.72 \times 10^{-5}$, is less than the one on the validation set, equals 0.0058, showing an excellent generalisation ability of the model.

%%%%%%%%%%%%%%%%%%%%%%%%%%%%%%%%%%%%%%%%%%%%%%%%%%%%%%%%%%%%%%%%%%%%%%%%%%

\section{Calibration Results}
\label{chap:calibration_results}
The trained DNN is ready to be used to execute calibration on the market data. The present section details how to get real market option quotes as well as the interest rate curve. We detail two scenarios: calibrating (i) five or (ii) three parameters and keeping known $v_0$ and a fixed $\kappa$. This section also compares the traditional calibration to the DNN-based one and presents how the trained neural network dramatically reduces Heston calibration’s computation time. We provide a realistic and thorough evaluation of the obtained results in sections~\ref{section:DE},~\ref{section:calib_5param} and~\ref{section:calib_3param}. \\

In this section, Heston model calibration is executed by minimising the set of five parameters $\theta, \sigma, \rho, \kappa, v_0$. We then fix $\kappa$ and $v_0$ and optimise parameters $\theta$, $\sigma$ and $\rho$. By this, the optimisation reduces to three parameters. This approach accelerates the calibration process by reducing the set of parameters to be optimised instead of minimising all five parameters.   

\subsection{Market Dataset Presentation}

\subsubsection{Market Option Quotes}
In order to perform the calibration on the market data, the first step is to get real market option prices.\\

There are several stock market indices, and the \emph{S\&P 500 index} is a benchmark. It is the abbreviation of the Standard \& Poor's 500 index, an index of the largest (by value) American companies listed on the NYSE (New York Stock Exchange), the world's largest stock exchange and the NASDAQ (National Association of Securities Dealers Automated Quotation System) the second largest stock market in the United States. The index includes 500 leading companies and reflects approximately 80\% of the market capitalization coverage available when counting Apple, Microsoft and Exxon as the three most prominent companies. It is considered the best representation of the US stock market and a barometer of the US economy. The code for the index called \emph{ticker} is \emph{SPX}.\\

In this work, we use the \emph{S\&P 500 index} options that we collect from \emph{yahoo finance} \url{https://finance.yahoo.com/quote/%5ESPX/options?p=%5ESPX}.
The platform offers call and put quotes. In our case, we only need the put table.
When calibrating, we consider put options with different maturities and expiries. For different maturities, there are different quotes for different strikes. Therefore, the data needed to perform calibration consists of maturities, strikes and the put prices corresponding to each strike and each maturity.\\ 

Put table provides maturities in the form of day/month/year, but we need it to be in years. To make the conversion, We use \emph{numpy.busday\_count} to count the number of days between collecting data and the options' expiry date. Since options are only traded on open days, we also remove US holidays to obtain the maturities in days that we divide by the number of trading days (equal to 252 days). Because the DNN has been trained for maturities between 0.05 and 20 years; we decide only to keep market data of maturities bigger than 0.05 years.\\

Furthermore, we convert strikes to integers.\\

There are bid and ask prices for put options: the bid price, the offer price, corresponds to the maximum price that a buyer is ready to pay immediately for financial security or an asset. The ask price, the asking price, is the minimum price a seller is willing to offer for the immediate sale of the same asset. A transaction occurs when the buyer and seller agree on the price of the asset in question, which must be lower than the bid price and higher than the ask price. Option prices that we consider for calibration are the mean between the bid and ask prices that are different from 0.\\

The market dataset we use dates from 10/08/2022 for a spot equal to $S_0 = 4210.24$. The file \emph{complete\_calib\_data.csv} contains the complete raw call and put data, including columns for prices as described above and maturities in years. The put data that we use in the calibration process is in the file \emph{puts\_data.csv}; it consists of $32$ different maturities from 0.0555 years to 5.3055 years, 489 different strikes from 200 to 9200 and 4625 market prices.   

\subsubsection{Interest Rate Curve}

Recall that calibration is the process of adjusting the parameters of a model such that model prices are compatible with market prices. We need an interest rate for each maturity to compute model prices. \emph{US Daily Treasury Par Yield Curve Rates}
\url{https://home.treasury.gov/policy-issues/financing-the-government/interest-rate-statistics?data=yield%27} provides interest rates for 12 maturities. On 10/08/2022, These rates are as follows:

\begin{table}[!htbp]
    \centering
    \resizebox{\textwidth}{!}
    {\begin{tabular}{|c|c|c|c|c|c|c|c|c|c|c|c|c|}
    \hline
Maturity & 1 Mo	& 2 Mo	& 3 Mo & 6 Mo & 1 Yr & 2 Yr	& 3 Yr	& 5 Yr & 7 Yr & 10 Yr & 20 Yr & 30 Yr\\
\hline
 Rate (\%) & 2.24 & 2.43 & 2.65 & 3.13 & 3.26 & 3.23 & 3.13 & 2.93 & 2.86 & 2.78 & 3.27 & 3.04 \\
\hline
    \end{tabular}}
    \caption{Daily Treasury Par Yield Curve Rates as on 10/08/2022}
    \label{tab:free_interest_rate}
\end{table}

To find rates corresponding to maturities that are not provided by \emph{US Daily Treasury Par Yield Curve Rates}, we perform interpolation in order to recover them. The Nelson Siegel Svensson method~\cite{nelson_siegel, svensson} is one of the most common methods of interpolating a curve. To implement it, we use the function \emph{calibrate\_ns\_ols} from \emph{nelson\-siegel\-svensson} package. Figure~\ref{fig:yield_curve} illustrates the extrapolated yield curve.

\begin{figure}[!htbp]
    \centering
    \includegraphics[width=12cm, height=6cm]{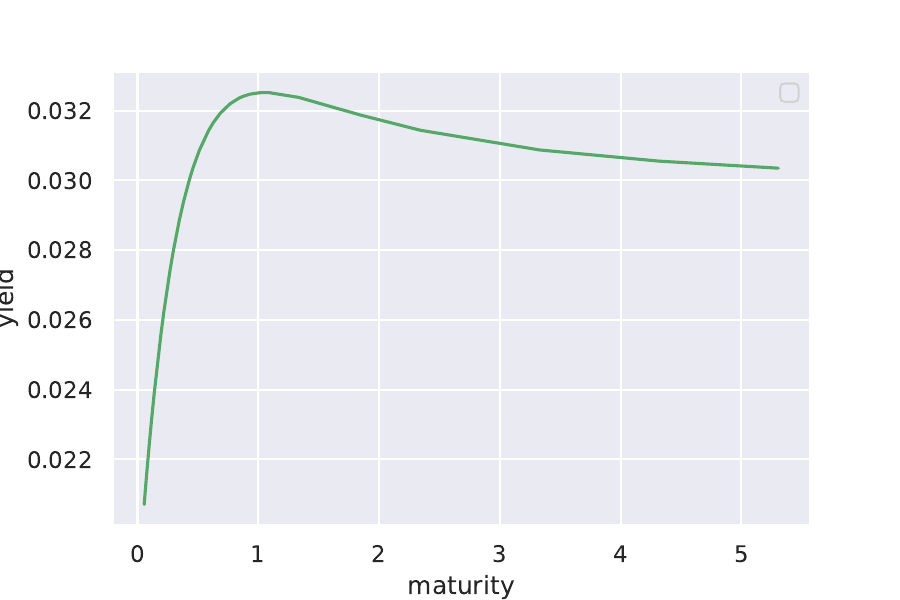}
    \caption{Yield curve}
    \label{fig:yield_curve}
\end{figure}

Note that for our calibration we use $32$ rates provided in the file \emph{rates.csv}.

\subsection{Optimising the Objective Function}

Remember, from section~\ref{section:model_calibration}, that model calibration consists in minimising the objective function that we choose to be the Weighted RMSE between the model prices $\hat{V}$ and the market ones $V_{Mkt}$ with respect to the set of parameters $\theta \in \Theta$ , given by:

\begin{equation*}
    J(\theta) = \sqrt{\sum_i \sum_j w_{i,j} \left(\hat{V}(\theta, T_i, m_j) - V_{Mkt}(T_i, m_j)\right)^2}
\end{equation*}

We opt for equal weights $w_{i,j}$ such that for all $i,j, w_{i,j} = 1/n$, where $n$ denotes the total available number of quotes. The put model prices used in the objective function are given by: 

\begin{equation*}
\label{equation:nfpp_def_1}
    P_t = K e^{-r(T-t)} \hat{P}_t
\end{equation*}

where $\hat{P}$ is the normalised forward put price.
\subsubsection{Differential Evolution Optimiser}
\label{section:DE}

As mentioned in section~\ref{section:traditional_methods}, the DE method's advantages include the algorithm's ability to find a global minimiser without any initial parameter set and without the need to compute gradients. However, the main disadvantage of this approach is that founding a global minimum is quite time-consuming. Firstly, we planned to apply the DE algorithm to calibrate the whole surface. When the model prices are calculated via semi-closed form formulas, the algorithm ran for more than 24 hours without returning a result. So we decided only to use this approach with the setting given in Table~\ref{tab:DE} as an initial experiment by examining the calibration on five parameters on the maturity having the least number of strikes, which is the last maturity (5.3055 years), for which 19 strikes, and then 19 market prices are quoted. Note that the search space of the DE algorithm is limited to the boundary chosen for each Heston model parameter given in Table~\ref{tab:sampling_range}. In order to implement the DE algorithm, we use the function \emph{differential\_evolution} from \emph{scipy.optimize} package.\\

\begin{table}[!htbp]
    \centering
    \begin{adjustbox}{width=0.5
    \textwidth}  
    \begin{tabular}{|c|c|}
    \hline
Parameter & setting\\
\hline
Strategy & best1bin\\
Number of generations & 90\\
Population size & 50\\
Tolerance for convergence & $10^{-6}$\\
Differential weight & 0.5\\
Crossover recombination & 0.7\\
\hline
    \end{tabular}
    \end{adjustbox}
    \caption{The setting of the Differential Evolution}
    \label{tab:DE}
\end{table}

For this specific case, the calibration process takes\footnote{Note that the CPU time may vary a little depending on the machine's current status. However, the ratio between the time of the traditional calibration and the DNN-based one remains pretty much the same.}5 minutes and 28 seconds when the model prices are calculated via DML, which is 16.5 times faster than the traditional method, which takes 1 hour, 30 minutes and 42 seconds, illustrating the benefit of the DML approach in reducing markedly the computation time. The achieved results are shown in Table~\ref{tab:diffevol_one_mat} where we display the calibrated parameters resulting from both the traditional and the DNN-based calibration as well as the absolute difference between them.
      
\begin{table}[!htbp]
    \centering
    \begin{adjustbox}{width=0.9\textwidth}  
    \begin{tabular}{|c|c|c|c|}
    \hline
Parameter & Traditional Calibration & DNN-based Calibration & Absolute Difference\\
\hline
$\theta$ & 0.0433 & 0.1353 &  0.0920\\
$\sigma$ & 0.3027 & 1.9980 &  1.6953\\
$\rho$ & -0.4290  & $-8.34 \times 10^{-7}$ & 0.4290 \\
$\kappa$ & 0.1580 & 1.3936 & 1.2356 \\
$v_0$ & 0.1017 & $1.02 \times 10^{-4}$ & 0.1016 \\
\hline
    \end{tabular}
    \end{adjustbox}
    \caption{Model calibration results using the differential evolution algorithm.}
    \label{tab:diffevol_one_mat}
\end{table}

Looking at the absolute difference in Table~\ref{tab:diffevol_one_mat}, it is clear that the calibrated parameters resulting from the DNN-based approach and those from the traditional method are not similar. This situation can be explained by the fact that the objective function to minimise does not reach the correct minima. Indeed, the DE algorithm needs a large number of generations, e.g. 1000, to give an accurate result, while we only use 90 generations. Increasing the number of generations will hugely increase the computational time when performing the traditional calibration. \\

As aforementioned, we do not optimise the objective function using the differential evolution optimiser to calibrate the Heston model parameters on the whole market dataset. We use the Nelder-Mead optimiser instead.

\subsubsection{Nelder-Mead Optimiser}
The Nelder-Mead algorithm is chosen to optimise the objective function. The algorithm's parameters are given in Table~\ref{tab:nelder_mead}:
                      
\begin{table}[!htbp]
    \centering
    \begin{adjustbox}{width=0.5\textwidth}  
    \begin{tabular}{|c|c|}
    \hline
Parameter & Value\\
\hline
Number of generations & 1000\\
Tolerance for convergence & $10^{-6}$\\
\hline
    \end{tabular}
    \end{adjustbox}
    \caption{Nelder-Mead parameters}
    \label{tab:nelder_mead}
\end{table}
The choice of the initial parameter set for the algorithm is crucial to find the correct optimal parameter set. We tried different initial parameter sets; the best initial guess we have fund when the calibration is performed on five parameters (respectively three parameters) is given in Table~\ref{tab:initial_guess_5params} (respectively ~\ref{tab:initial_guess_3params}).

\begin{table}[!htbp]
    \centering
     \begin{adjustbox}{width=0.9\textwidth}  
    \begin{tabular}{|c|c c c c c|}
    \hline
Parameter & $\theta$ & $\sigma$ & $\rho$ & $\kappa$ & $v_0$\\
\hline
Value & 0.1021 & 1.5986 & -0.3899 & 1.4719 &  $1.12 \times 10^{-5}$ \\
\hline
    \end{tabular}
    \end{adjustbox}
    \caption{Initial parameter set - calibration on five parameters.}
    \label{tab:initial_guess_5params}
\end{table}

\begin{table}[!htbp]
    \centering
    \begin{adjustbox}{width=0.6\textwidth}  
    \begin{tabular}{|c|c c c|}
    \hline
Parameter & $\theta$ & $\sigma$ & $\rho$ \\
\hline
Value & 0.2752 & 0.4571 & -0.4477 \\
\hline
    \end{tabular}
    \end{adjustbox}
    \caption{Initial parameter set - calibration on three parameters}
    \label{tab:initial_guess_3params}
\end{table}

Note that to implement the Nelder-Mead algorithm, we use the function \emph{fmin} from \emph{scipy.optimize} package.

\subsection{Five Parameters Calibration}
\label{section:calib_5param}

Now, we perform a calibration on the whole surface on the Heston pricer and its DML equivalent. The optimal parameters that we obtain and the computation times of both approaches are presented in Table~\ref{tab:calib_5params}.

\begin{table}[!htbp]
    \centering
    \resizebox{\textwidth}{!}
    {\begin{tabular}{|c|c|c|c|c|c|c|}
    \hline
Parameter & $\theta$ & $\sigma$ & $\rho$ & $\kappa$ & $v_0$ & CPU Time\\
\hline
Traditional Calibration & 0.1072 & 1.5986 & -0.3899 & 1.4719 & $1.12 \times 10^{-5}$ &  00:20:23 \\
DNN-based Calibration & 0.1072 & 1.5986 & -0.3899 & 1.4719 & $1.12 \times 10^{-5}$ & 00:00:20\\
\hline
    \end{tabular}}
    \caption{Model calibration on five parameters results.}
    \label{tab:calib_5params}
\end{table}

It can be seen from Table~\ref{tab:calib_5params} that the calibrated parameters resulting from DML are perfect fit for those resulting from the Heston pricer. The DNN-based calibration takes only 20 seconds, which is 61 times faster than the traditional method. This considerably time saving is explained by the fact that the DNN price is considered "deterministic", as stated in section~\ref{section: ANN_approach}, whilst the Heston price involves an integral approximation that takes longer to compute.\\

Let us now look at the performance of the two sets of calibrated parameters in computing the option prices. To this end, both sets of calibrated parameters are inserted into the traditional Heston pricer to compute both Heston prices. To be more specific:

\begin{enumerate}
 \item  We insert the optimal parameters resulting from the traditional calibration into the traditional Heston pricer and compute Heston model prices. We refer to this price as \emph{Heston price}.
 
 \item  We insert the optimal parameters resulting from the DNN-based calibration into the traditional Heston pricer and compute Heston model prices. We refer to this price as \emph{DNN price}.
\end{enumerate}
 
We then compare Heston and DNN prices with the market prices for a maturity equal to 5.3055 years. Figure~\ref{fig:prices_5param} illustrates both model prices as a function of the market price. Figure~\ref{fig:diff_prices_5param} displays Heston, DNN and market prices as a function of the strike.

\begin{figure}[!htbp]
    \centering
    \includegraphics[width=14cm, height=8cm]{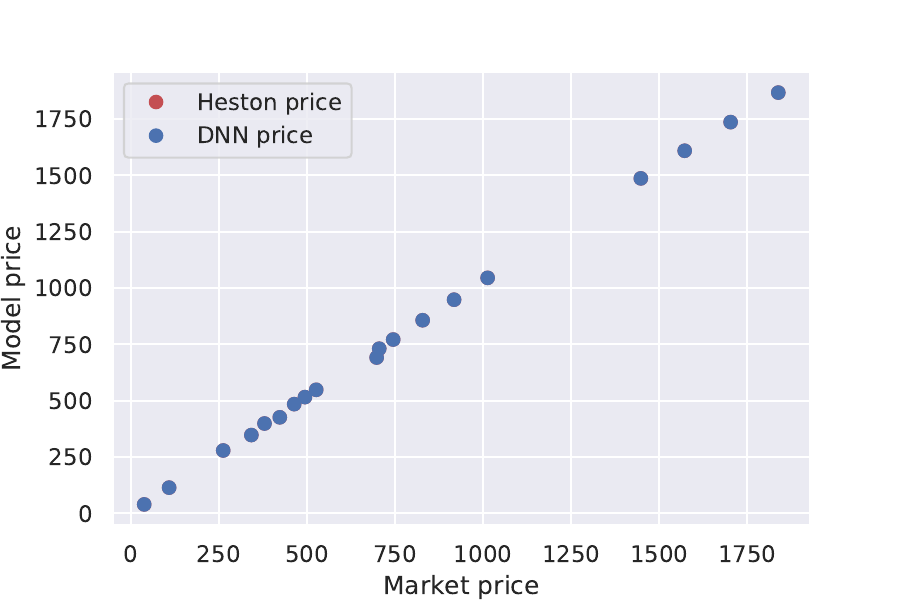}
    \caption{Option price comparison for traditional and DNN-based calibration parameters and market prices. Here the calibration is based on five parameters. We do not see Heston prices because Heston and DNN prices overlap since they are equal.}
    \label{fig:prices_5param}
\end{figure}

\begin{figure}[!htbp]
    \centering
    \includegraphics[width=14cm, height=8cm]{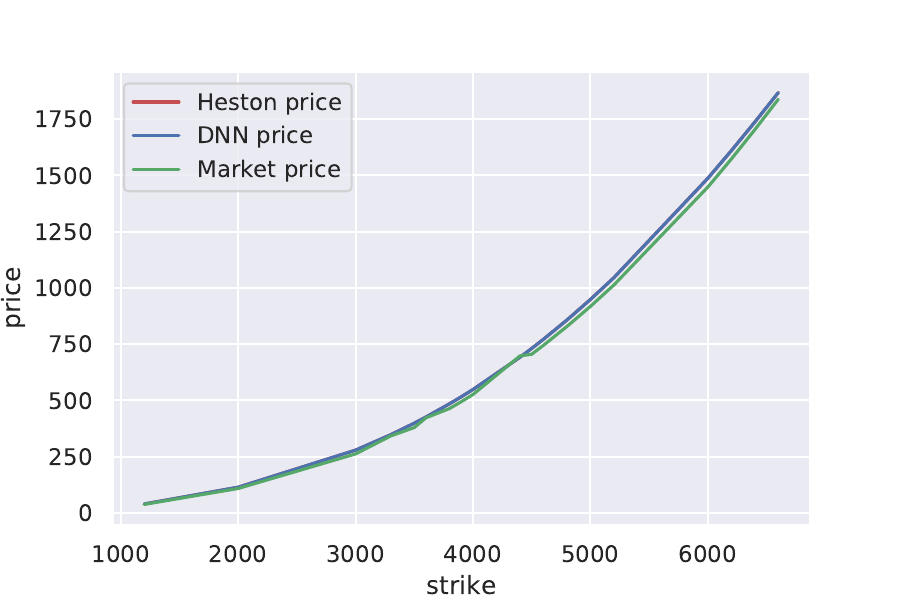}
    \caption{Market, Heston and DNN prices. Here the calibration is based on five parameters. We do not see Heston prices because Heston and DNN prices overlap since they are equal.}
    \label{fig:diff_prices_5param}
\end{figure}

Looking at the graphs, we notice that the prices computed via the DNN and the Heston pricer are equal, which is expected since their calibrated parameters are the same. This confirms the capacity of fitting of the differential network with respect to the Heston pricer. Moreover, both model prices are a good fit for market prices since they are very close to market prices, and the average absolute error between the model and market prices is relatively small; it equals 21.52 for market prices ranging from 38.10 to 1837.85.\\

According to these results, the calibration accuracy of the DNN and the Heston pricer is the same. Nevertheless, there is a vast difference between their computational times, which substantially highlights the advantage of using a trained DNN instead of traditional Heston calibration methods.

\subsection{Three Parameters Calibration}
\label{section:calib_3param}

In this section, $\kappa$ and $v_0$ are fixed. This reduces the optimisation to three parameters, namely $\theta$, $\sigma$ and $\rho$. \\

We choose $\kappa = 0.15$. The ATM implied volatility (IV) of an option with the shortest maturity can be used to approximate $v_0: v_0 \approx \sigma^{ATM^2}_{imp}$. Let us then find the ATM IV of the shortest maturity:  We know that the ATM IV is the IV corresponding to moneyness 0. The S\&P 500 index data of the 10 August 2022 \url{https://www.barchart.com/stocks/quotes/$SPX/options?expiration=2022-08-11-w&moneyness=allRows} does not contain a moneyness 0, we then make an interpolation between the two IV values corresponding to the two closest moneyness to 0. Thus the ATM IV is the mean between 1.5\% and 2.04\% : 

\begin{equation*}
  \sigma_{imp}^{ATM} =   \frac{1.5 \% + 2.04 \%}{2} = 0.0177
\end{equation*}

Thus 

\begin{equation*}
    v_0 = \sigma^{ATM^2}_{imp} = 3.13 \times 10^{-4}
\end{equation*}

As with five parameters calibration, we now  perform a calibration on the whole surface on the Heston pricer and its DML equivalent. Table~\ref{tab:calib_3params} describes the optimal parameters that we obtain and the computation time.

\begin{table}[!htbp]
    \centering
    \begin{adjustbox}{width=0.9\textwidth}  
    \begin{tabular}{|c|c|c|c|c|}
    \hline
Parameter & $\theta$ & $\sigma$ & $\rho$ & CPU Time\\
\hline
Traditional calibration & 0.2890 & 0.4571 & -0.4477 & 00:08:14 \\
DNN-based calibration & 0.2890 & 0.4571 & -0.4477 & 00:00:14\\
\hline
    \end{tabular}
    \end{adjustbox}
    \caption{Model calibration on three parameters results.}
    \label{tab:calib_3params}
\end{table}

It can be seen from Table~\ref{tab:calib_3params} that the calibrated parameters resulting from DML are perfect fit for those resulting from the Heston pricer. The DNN-based calibration takes only 14 seconds, which is 35 times faster than the traditional method. As aforementioned, this considerably time-saving is explained by the fact that the DNN price is considered "deterministic", whilst the Heston price involves an integral approximation that takes longer to compute. In addition, we notice that the traditional calibration on three parameters is almost 2.5 times faster than that on five. Also, the DNN-based calibration on three parameters takes almost 1.5 times faster than that on five and 87 times faster than the traditional calibration on five parameters. This is expected as the number of optimal parameters is reduced. \\

Like in section~\ref{section:calib_5param} we look at the performance of the two sets of calibrated parameters in computing the option prices. To this end, both sets of calibrated parameters are inserted into the traditional Heston pricer to compute both Heston prices. We then compare both model Heston prices with the market prices for a maturity equal to 5.3055 years. Figure~\ref{fig:prices_3param} illustrates both model prices as a function of the market price. Figure~\ref{fig:diff_prices_3param} displays Heston, DNN and market prices as a function of the strike.\\

\begin{figure}[!htbp]
    \centering
    \includegraphics[width=14cm, height=8cm]{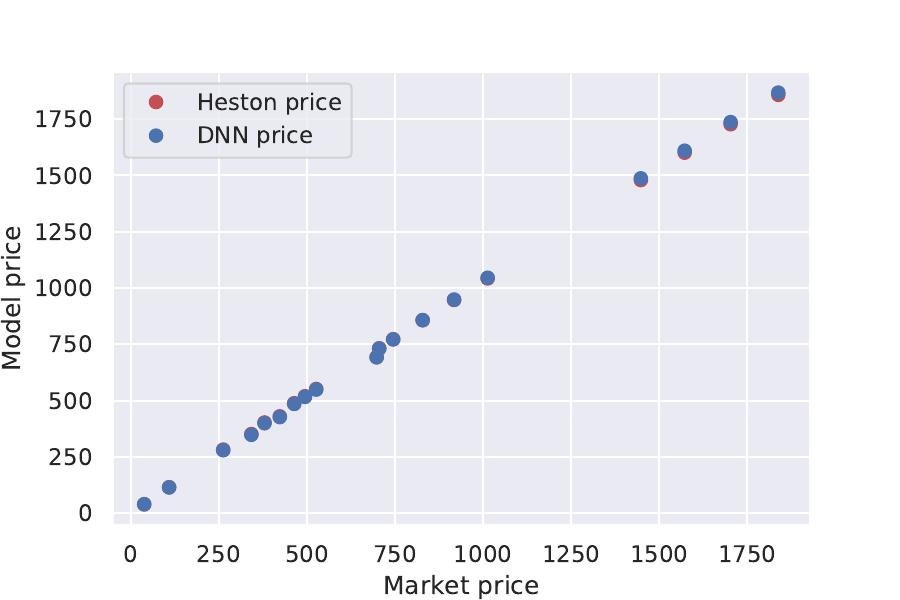}
    \caption{Option price comparison for traditional and DNN-based calibration parameters and market prices. Here the calibration is based on three parameters. We do not see Heston prices because Heston and DNN prices overlap since they are equal.}
    \label{fig:prices_3param}
\end{figure}

\begin{figure}[!htbp]
    \centering
    \includegraphics[width=14cm, height=8cm]{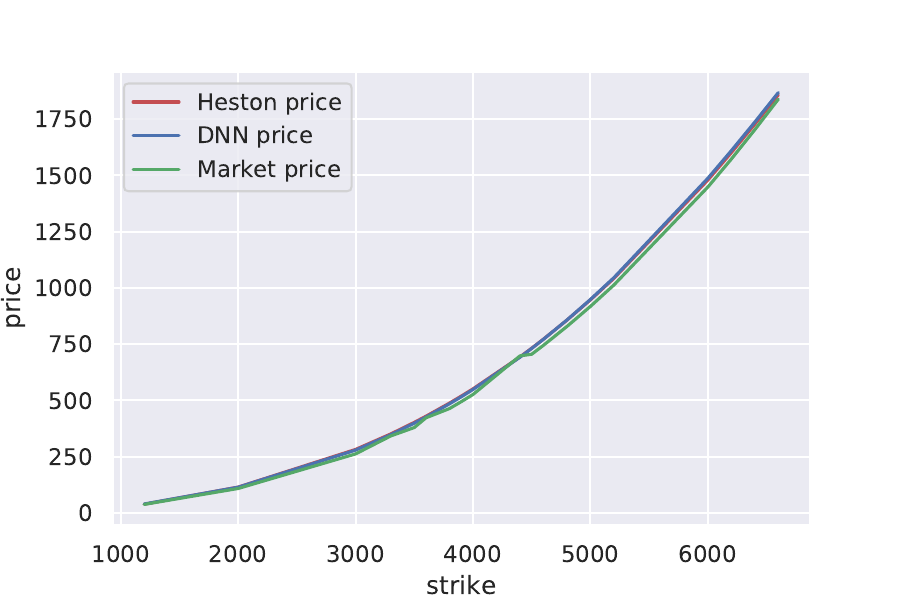}
    \caption{Market, Heston and DNN prices. Here the calibration is based on three parameters. We do not see Heston prices because Heston and DNN prices overlap since they are equal.}
    \label{fig:diff_prices_3param}
\end{figure}

Looking at the graphs, we notice that the prices resulting from DNN and Heston pricer are equal, which is expected since their calibrated parameters are the same. This confirms the capacity of fitting of the differential network with respect to the Heston pricer. Moreover, both model prices are a good fit for market prices since they are very close to market prices, and the average absolute error between the model and market prices is relatively small; it equals 20.46 for market prices ranging from 38.10 to 1837.85 and is close to the one calculated on five parameters (equals 21.52). These results imply that the prices resulting from both the DNN-based and the traditional calibration on three parameters are very similar to those on five.\\

According to these results, the DNN and the Heston pricer's calibration accuracy is the same and similar to the one on five parameters. However, the same results can be achieved much faster when the calibration is performed on three parameters with the DNN-based approach, which implies that using a trained DNN on three parameters instead of traditional Heston calibration methods on five parameters is more advantageous.

\section{Conclusion and Further Research}
\label{chap:conclusion}

This project's overarching contribution consists of applying the DML technique to price vanilla European options (i.e. the calibration instruments) when the underlying asset follows a Heston model and then calibrating the model using the trained network.

\subsection{Findings and Contributions}

In this work, we introduced different regularisation techniques applied to the DML. We compared their performance in reducing overfitting and improving the generalisation error. Furthermore, the DML performance was compared to the classical DL (without differentiation) one. As a result, we showed that the DML vastly outperforms the DL. We also demonstrated that DML could learn effectively from small datasets compared to DL. Thus training is fast (fewer labels are needed), and the pricing is accurate.\\

In this work's last part, we calibrated the market data in two scenarios, calibrating either five or three parameters. We compared the traditional calibration to the DNN-based one for both options. We found that the calibration accuracy with three parameters is similar to that with five. The calibration of three parameters is faster regardless of the price computation type (real pricer or DNN-based). However,  the results can be achieved much faster when the calibration is performed with the DNN-based approach. As a result, using the DNN-based price and calibrating three parameters dramatically reduces Heston calibration’s computation time since it predicts prices considerably faster.

%%%%%%%%%%%%%

\subsection{Further Research}
This project focuses on the vanilla European option pricing because our goal was to calibrate the Heston model on the market data. Nevertheless, the pricing of other
type of options, for example the Asian, Barrier and some other European exotic options, using the DML technique can be performed in the future. The challenge in this case is to compute real option prices. One can use finite difference or MC techniques to produce the labels. Another possibility is to consider directly the payoff as a label as in ~\cite{differential_ML}. \\

Considering the payoff as a label leads to a question potentially being studied in further works. Is the DNN's accuracy higher when the label is the price (compared to considering the payoffs) as we may intuit? When using the payoff as a label in the approach presented in this work, it can be used for other dynamics (e.g. local volatility moels, stochastic local volatility models, etc.).\\

%%%%%%%%%%%%%%%%%%%%%%%%%%%%%%%%%%%%%%%%%%%%%%%%%%%%%%%%%%%%%%%%%%%%%%%%%%%
%%%%%%%%%%%%%%%%%%%%%%%%%%%%%%%%%%%%%%%%%%%%%%%%%%%%%%%%%

\newpage

\appendix

\section{Proof of Proposition \ref{proposition:bs_particular_case}}
\label{section:dem_bs_particular_case}

Itô's lemma applied to the log-moneyness forward $\hat{F}$ gives:
\begin{equation*}
    d\hat{F}_{t,T} = (r-\frac{1}{2} v_t)dt + \sqrt{v_t}dW^1_t
\end{equation*}

When $\sigma << \epsilon$, Heston variance (\ref{equation:heston_sde}) satisfies:
\begin{equation*}
    dv_t = \kappa (\theta-v_t)dt
\end{equation*}
\begin{equation*}
  \Rightarrow  e^{\kappa t}\left(dv_t + \kappa v_t dt \right) = e^{\kappa t}\kappa \theta dt
\end{equation*}

\begin{equation*}
  \Rightarrow  e^{\kappa T} v_T = e^{\kappa t} v_t + {\displaystyle \int^{T}_t} e^{\kappa s} \kappa \theta ds
  \end{equation*}
 
  \begin{equation*}
  \begin{array}{cl}
    \Rightarrow v_T  & = e^{-\kappa (T-t)} v_t  + {\displaystyle \int^{T}_t} e^{-\kappa (T-s)} \kappa \theta ds\\
      & = e^{-\kappa (T-t)} v_t  + \kappa \theta e^{-\kappa T}\left[\frac{1}{\kappa} e^{ks}\right]^T_t
   \end{array}
   \end{equation*}
   which gives:
   \begin{equation*}
       v_T = \theta + e^{-\kappa (T-t)} (v_t - \theta)
   \end{equation*}
or:
\begin{equation}
       \label{equation:nu}
       v_t = \theta + e^{-\kappa t} (v_0 - \theta)
   \end{equation}
The implied volatility is defined as:
\begin{equation}
\label{equation:sigma_implicite}
    \sigma^2_{imp} T = {\displaystyle \int^{T}_0} v_s ds
\end{equation}
By substituting (\ref{equation:nu}) into (\ref{equation:sigma_implicite}) we obtain:

\begin{eqnarray*}
\begin{array}{cl}
\sigma^2_{imp} T & = {\displaystyle \int^{T}_0} \theta + e^{-\kappa s}(v_0-\theta) ds \\
 & = \theta T + (v_0-\theta) \left[-\frac{1}{\kappa} e^{-\kappa s} \right]^T_0 \\
  & = \theta T + \frac{1-e^{-\kappa T}}{\kappa} (v_0-\theta)
\end{array}
\end{eqnarray*}
 which gives the result.
%%%%%%%%%%%%%%%%%%%%%%%%%%%%%%%%%%%%%%%%%%%%%%%%%%%%%%%
\section{Proof of the Different Equations in Section~\ref{section:diff_p_prime_inputs}}
\label{appendix:appendix_b}
\subsection{Proof of Equation (\ref{equation:diff_wrt_theta})}
\label{section:demo_diff_wrt_theta}

Calculating the partial derivative of the normalised forward put price $\hat{P}$ with respect to $\theta$, we get:
\begin{equation}
    \frac{\partial \hat{P}}{\partial \theta} =  - \frac{1}{\pi} {\displaystyle \int^{+ \infty}_0 Re \left[e^{ (i u + \frac{1}{2})\hat{F}_{t,T}} \  \frac{\partial \phi_{\tau}\left(u-\frac{i}{2}\right)}{\partial \theta}\right]\frac{\,du}{u^2+\frac{1}{4}}}
\end{equation}

We then differentiate the characteristic function (\ref{equation:phi}) with respect to $\theta$:
\begin{equation}
\label{equation:diff_phi_theta}
    \frac{\partial \phi_{\tau}(u)}{\partial \theta} = \phi_{\tau}(u)\left[\frac{\partial A(u)}{\partial \theta} +v_0 \frac{\partial B(u)}{\partial \theta}\right]
\end{equation}
We have:
\begin{equation}
\label{equation:diff_A_theta}
\frac{\partial A(u)}{\partial \theta} = \frac{A(u)}{\theta} 
\end{equation}
 and:
 \begin{equation}
     \label{equation:diff_B_theta}
    \frac{\partial B(u)}{\partial \theta} = 0
 \end{equation}

By substituting (\ref{equation:diff_A_theta}) and (\ref{equation:diff_B_theta}) in (\ref{equation:diff_phi_theta}), we obtain:

\begin{equation*}
    \frac{\partial \phi_{\tau}(u)}{\partial \theta} = \frac{A(u)}{\theta} \ \phi_{\tau}(u)
\end{equation*}
which yields the result.
%%%%%%%%%%%%%%%%%%%%%%%%%%%%%%%%%%%%%%%%%%%%%%
\subsection{Proof of Equation (\ref{equation:diff_wrt_m})}
\label{section:demo_diff_wrt_m}

Calculating the partial derivative of the normalised forward put price $\hat{P}$ with respect to  $m$, we get:
\begin{equation}
\label{equation:diff_p_prime_m}
    \frac{\partial \hat{P}}{\partial m} =  - \frac{1}{\pi} {\displaystyle \int^{+ \infty}_0 Re \left[e^{ (i u + \frac{1}{2})r \tau}  \phi_{\tau}\left(u-\frac{i}{2}\right)   \frac{\partial e^{ (i u + \frac{1}{2}) m}}{\partial m}\right]\frac{\,du}{u^2+\frac{1}{4}}}
\end{equation}

We have:
\begin{equation}
\label{equation:diff_m}
\frac{\partial e^{ (i u + \frac{1}{2}) m}}{\partial m} =  (i u + \frac{1}{2})\  e^{ (i u + \frac{1}{2}) m}
\end{equation}

By substituting (\ref{equation:diff_m}) in (\ref{equation:diff_p_prime_m}), we obtain the result.

%%%%%%%%%%%%%%%%%%%%%%%%%%%%%%%%%%%%
\subsection{Proof of Equation (\ref{equation:diff_wrt_v0})}
\label{section:demo_diff_wrt_v0}

Calculating the partial derivative of the normalised forward put price $\hat{P}$ with respect to  $v_0$, we get:

\begin{equation}
\label{equation:diff_p_prime}
    \frac{\partial \hat{P}}{\partial v_0} =  - \frac{1}{\pi} {\displaystyle \int^{+ \infty}_0 Re \left[e^{ (i u + \frac{1}{2})\hat{F}_{t,T}} \ \frac{\partial \phi_{\tau}\left(u-\frac{i}{2}\right)}{\partial v_0}\right]\frac{\,du}{u^2+\frac{1}{4}}}
\end{equation}

Differentiating the characteristic function (\ref{equation:phi}) with respect to $v_0$ yields:

\begin{equation}
\label{equation:diff_phi_v0}
    \frac{\partial \phi_{\tau}(u)}{\partial v_0} = B(u)\phi_{\tau}(u)
\end{equation}

By substituting (\ref{equation:diff_phi_v0}) in (\ref{equation:diff_p_prime}), we get the result.
%%%%%%%%%%%%%%%%%%%%%%%%%%%%%%%%%%%%

\subsection{Proof of Equation (\ref{equation:diff_wrt_r})}
\label{section:demo_diff_wrt_r}

Calculating the partial derivative of the normalised forward put price $\hat{P}$ with respect to $r$, we get:
\begin{equation}
\label{equation:diff_p_prime_r}
    \frac{\partial \hat{P}}{\partial r} =  - \frac{1}{\pi} {\displaystyle \int^{+ \infty}_0 Re \left[e^{ (i u + \frac{1}{2})m} \ \phi_{\tau}(u-\frac{i}{2}) \  \frac{\partial e^{ (i u + \frac{1}{2}) r \tau}}{\partial r}\right]\frac{\,du}{u^2+\frac{1}{4}}}
\end{equation}

We have:
\begin{equation}
\label{equation:diff_r}
\frac{\partial e^{ (i u + \frac{1}{2}) r \tau}}{\partial r} =  \tau \ (i u + \frac{1}{2}) \ e^{ (i u + \frac{1}{2}) r \tau}
\end{equation}

A substitution of (\ref{equation:diff_r}) in (\ref{equation:diff_p_prime_r}) and the expression of the partial derivative of the normalised forward put price $\hat{P}$ with respect to  the log moneyness $m$ (\ref{equation:diff_wrt_m}) give the result.

%%%%%%%%%%%%%%%%%%%%%%%%%%%%%%%%%%%%%%%%%%%%%%%%%%%%%

\subsection{Proof of Equation (\ref{equation:diff_wrt_tau})}
\label{section:demo_diff_wrt_tau}

The partial derivative of the normalised forward put price $\hat{P}$ with respect to the time to maturity $\tau$ is:

\begin{equation*}
    \frac{\partial \hat{P}}{\partial \tau} =  - \frac{1}{\pi} {\displaystyle \int^{+ \infty}_0 Re \left \{ \frac{\partial}{\partial \tau} \left[e^{ (i u + \frac{1}{2}) \hat{F}_{t,T}} \phi_{\tau}(u-\frac{i}{2}) \right] \right\}\frac{\,du}{u^2+\frac{1}{4}}} 
\end{equation*}   

We have:
\begin{equation*}
  \frac{\partial}{\partial \tau} \left[e^{ (i u + \frac{1}{2}) \hat{F}_{t,T}} \phi_{\tau}(u-\frac{i}{2}) \right] =   r e^{ (i u + \frac{1}{2})\hat{F}_{t,T}} \ (i u + \frac{1}{2}) \ \phi_{\tau}(u-\frac{i}{2}) 
    + e^{ (i u + \frac{1}{2})\hat{F}_{t,T}} \ \frac{\partial \phi_{\tau}(u-\frac{i}{2})}{\partial \tau}
\end{equation*}

Then:

 \begin{eqnarray*}
 \frac{\partial \hat{P}}{\partial \tau} & =   - \frac{1}{\pi} \biggl\{ r {\displaystyle \int^{+ \infty}_0 Re \left[e^{ (i u + \frac{1}{2})\hat{F}_{t,T}} (i u + \frac{1}{2}) \  \phi_{\tau}(u-\frac{i}{2}) \right]\frac{\,du}{u^2+\frac{1}{4}}} \\
   & + {\displaystyle \int^{+ \infty}_0 Re \left[e^{ (i u + \frac{1}{2})\hat{F}_{t,T}} \  \frac{\partial \phi_{\tau}(u-\frac{i}{2})}{\partial \tau}\right] \frac{\,du}{u^2+\frac{1}{4}}} \biggl\} \nonumber
    \end{eqnarray*}  

Differentiating the characteristic function (\ref{equation:phi}) with respect to $\tau$ yields:

\begin{equation}
\label{equation:phi_to}
    \frac{\partial \phi_{\tau}(u)}{\partial \tau} = \phi_{\tau}(u) \left[
    \frac{\partial A(u)}{\partial \tau} + v_0 \frac{\partial B(u)}{\partial \tau} \right]
\end{equation}

Long calculations lead to:
\begin{equation}
    \label{equation: diff_A_tau}
    \frac{\partial A(u)}{\partial \tau} = \frac{\kappa \theta}{\sigma^2} \left[\beta - d + 2 \ \frac{d g e^{-d\tau}}{g e^{-d\tau} -1} \right]
\end{equation}

and:

\begin{equation}
    \label{equation: diff_B_tau}
    \frac{\partial B(u)}{\partial \tau} = \frac{\beta - d}{\sigma^2} \ \frac{d e^{-d\tau} (1-g)}{(1-g e^{-d\tau})^2}
\end{equation}

A substitution of (\ref{equation: diff_A_tau}) and (\ref{equation: diff_B_tau}) in (\ref{equation:phi_to}) leads to (\ref{equation:diff_phi_tau}).

%%%%%%%%%%%%%%%%%%%%%%%%%%%%%%%%%%%%%%%%%%%%%%%%%%%%%

%%%%%%%%%%%%%%%%%%%%%%%%%%%%%%%%%%%%%%%%%%%%%%%%%%%%%

\subsection{Proof of Equation (\ref{equation:diff_wrt_kappa})}
\label{section:demo_diff_wrt_kappa}

Calculating the partial derivative of the normalised forward put price $\hat{P}$ with respect to the reversion speed of variance process $\kappa$, we get:
\begin{equation*}
    \frac{\partial \hat{P}}{\partial \kappa} =  - \frac{1}{\pi} {\displaystyle \int^{+ \infty}_0 Re \left[ e^{ (i u + \frac{1}{2}) \hat{F}_{t,T}} \ \frac{\partial \phi_{\tau}(u-\frac{i}{2})}{\partial \kappa}\right]\frac{\,du}{u^2+\frac{1}{4}}}
\end{equation*}

Differentiating the characteristic function (\ref{equation:phi}) with respect to $\kappa$ yields:

\begin{equation*}
    \frac{\partial \phi_{\tau}(u)}{\partial \kappa} = \phi_{\tau}(u) \left[
    \frac{\partial A(u)}{\partial \kappa} + v_0 \frac{\partial B(u)}{\partial \kappa} \right]
\end{equation*}

We have:
\begin{eqnarray}
    \label{equation:diff_A_kappa_initial_form}
    \frac{\partial A(u)}{\partial \kappa} & = \ \frac{\theta}{\sigma^2} \left[\tau (\beta-d)  - 2\ln \left(\frac{g e^{-d\tau} - 1}{g-1}\right) \right] \\
     & + \frac{\kappa \theta}{\sigma^2} \left \{ \tau \ \frac{\partial \beta}{\partial \kappa}  - \tau \ \frac{\partial d}{\partial \kappa} - 2 \ \frac{\partial}{\partial \kappa} \left [\ln \left(\frac{g e^{-d\tau} - 1}{g-1}\right) \right] \right \} \nonumber
\end{eqnarray}

Differentiating $ \ln \left(\frac{g e^{-d\tau} - 1}{g-1}\right)$ with respect to $\kappa$ requires the calculation of the partial derivative of $\beta$, $d$ and $g$ with respect to $\kappa$ that gives:

\begin{equation}
    \label{equation:diff_beta_kappa}
    \frac{\partial \beta}{\partial \kappa} = 1
\end{equation}

\begin{equation}
    \label{equation:diff_d_kappa}
    \frac{\partial d}{\partial \kappa} = \frac{\beta}{d}
\end{equation}

\begin{equation}
    \label{equation:diff_g_kappa}
    \frac{\partial g}{\partial \kappa} = -2 \  \frac{g}{d}
\end{equation}

 yielding after further calculation to:
 \begin{equation}
     \label{equation:diff_ln_kappa}
     \frac{\partial}{\partial \kappa} \left[\ln \left(\frac{g e^{-d\tau} - 1}{g-1}\right) \right] = \frac{2g (g e^{-d\tau}-1) - g (g-1) e^{-d\tau}(2+\tau \beta) }{d (g-1) (g e^{-d\tau}-1)}
 \end{equation}
 
By substituting (\ref{equation:diff_beta_kappa}), (\ref{equation:diff_d_kappa}) and (\ref{equation:diff_ln_kappa}) in (\ref{equation:diff_A_kappa_initial_form}) and after  simplification, we obtain (\ref{equation:diff_A_kappa}). Moreover, the partial derivative of $B$ with respect to $\kappa$ is:

\begin{equation}
\label{diff_B_kappa_initial_form}
    \frac{\partial B(u)}{\partial \kappa} = \frac{1}{\sigma^2} \ \frac{1-e^{-d\tau}}{1-g e^{-d\tau}} \left(\frac{\partial \beta}{\partial \kappa} - \frac{\partial d}{\partial \kappa} \right) + \frac{\beta-d}{\sigma^2} \  \frac{\partial}{\partial \kappa} \left( \frac{1-e^{-d\tau}}{1-g e^{-d\tau}}\right)
\end{equation}

Long calculations lead to:
\begin{equation}
     \label{equation:diff_quot_kappa}
     \frac{\partial}{\partial \kappa} \left( \frac{1-e^{-d\tau}}{1-g e^{-d\tau}}\right) =  \frac{1}{d} \ \left( \frac{\tau \beta e^{-d\tau}}{1- g e^{-d\tau}}  - \frac{(1-e^{-d\tau}) \ g \ e^{-d\tau} \ (2+\tau \beta)}{(1- g e^{-d\tau})^2} \right)
 \end{equation}

Substituting (\ref{equation:diff_beta_kappa}), (\ref{equation:diff_d_kappa}) and (\ref{equation:diff_quot_kappa}) in (\ref{diff_B_kappa_initial_form}) and after simplification, we get (\ref{equation:diff_B_kappa}).

%%%%%%%%%%%%%%%%%%%%%%%%%%%%%%%%%%%%%%%%%%%%%%%%%%%%%

%%%%%%%%%%%%%%%%%%%%%%%%%%%%%%%%%%%%%%%%%%%%%%%%%%%%%

\subsection{Proof of Equation (\ref{equation:diff_wrt_rho})}
\label{section:demo_diff_wrt_rho}

Calculating the partial derivative of the normalised forward put price $\hat{P}$ with respect to the  correlation between the variance and the underlying process $\rho$, we get:

\begin{equation*}
    \frac{\partial \hat{P}}{\partial \rho} =  - \frac{1}{\pi} {\displaystyle \int^{+ \infty}_0 Re \left[ e^{ (i u + \frac{1}{2}) \hat{F}_{t,T}} \ \frac{\partial \phi_{\tau}(u-\frac{i}{2})}{\partial \rho}\right]\frac{\,du}{u^2+\frac{1}{4}}}
\end{equation*}

Differentiating the characteristic function (\ref{equation:phi}) with respect to $\rho$ yields:

\begin{equation*}
    \frac{\partial \phi_{\tau}(u)}{\partial \rho} = \phi_{\tau}(u) \left[
    \frac{\partial A(u)}{\partial \rho} + v_0 \frac{\partial B(u)}{\partial \rho} \right]
\end{equation*}

We have:
\begin{equation}
    \label{equation:diff_A_rho_initial_form}
    \frac{\partial A(u)}{\partial \rho}  =  \frac{\kappa \theta}{\sigma^2} \left \{\tau \left (\frac{\partial \beta}{\partial \rho} - \frac{\partial d}{\partial \rho} \right)
    - 2 \ \frac{\partial}{\partial \rho} \left [\ln \left(\frac{g e^{-d\tau} - 1}{g-1}\right) \right] \right \}
\end{equation}

Differentiating $ \ln \left(\frac{g e^{-d\tau} - 1}{g-1}\right)$ with respect to $\rho$ requires the calculation of the partial derivative of $\beta$, $d$ and $g$ with respect to $\rho$ that gives:

\begin{equation}
    \label{equation:diff_beta_rho}
    \frac{\partial \beta}{\partial \rho} = - i\ u \ \sigma
\end{equation}

\begin{equation}
    \label{equation:diff_d_rho}
    \frac{\partial d}{\partial \rho} = \frac{- i u \sigma \beta}{d}
\end{equation}

\begin{equation}
    \label{equation:diff_g_rho}
    \frac{\partial g}{\partial \rho} = \frac{2 i u \sigma g}{d}
\end{equation}

 yielding after further calculation to:
 \begin{equation}
     \label{equation:diff_ln_rho}
     \frac{\partial}{\partial \rho} \left[\ln \left(\frac{g e^{-d\tau} - 1}{g-1}\right) \right] = \ \frac{iu\sigma g}{d} \ \left[ \frac{e^{-d\tau} (2 + \tau \beta)} {g e^{-d\tau}-1} - \frac{2}{g-1} \right]
 \end{equation}
 
By substituting (\ref{equation:diff_beta_rho}), (\ref{equation:diff_d_rho}) and (\ref{equation:diff_ln_rho}) in (\ref{equation:diff_A_rho_initial_form}) and after  simplification, we obtain (\ref{equation:diff_A_rho}). Moreover, the partial derivative of $B$ with respect to $\rho$ is:

\begin{equation}
\label{equation:diff_B_rho_initial_form}
    \frac{\partial B(u)}{\partial \rho} = \frac{1}{\sigma^2} \left(\frac{\partial \beta}{\partial \rho} - \frac{\partial d}{\partial \rho} \right) \left( \frac{1-e^{-d\tau}}{1-g e^{-d\tau}}\right) + \frac{\beta - d}{{\sigma}^2} \ \frac{\partial}{\partial \rho} \left( \frac{1-e^{-d\tau}}{1-g e^{-d\tau}}\right)
\end{equation}

Long calculations lead to:
\begin{equation}
     \label{equation:diff_quot_rho}
     \frac{\partial}{\partial \rho} \left( \frac{1-e^{-d\tau}}{1-g e^{-d\tau}}\right) =  \frac{iu\sigma e^{-d\tau}}{d \ (1-g e^{-d\tau})^2}  \left[ -\tau \beta(1-g e^{-d\tau}) + g(2+\tau \beta) (1-e^{-d\tau}) \right]
 \end{equation}

Substituting (\ref{equation:diff_beta_rho}), (\ref{equation:diff_d_rho}) and (\ref{equation:diff_quot_rho}) in (\ref{equation:diff_B_rho_initial_form}) and after simplification, we get (\ref{equation:diff_B_rho}).

%%%%%%%%%%%%%%%%%%%%%%%%%%%%%%%%%%%%%%%%%%%%%%%%%%%%%
%%%%%%%%%%%%%%%%%%%%%%%%%%%%%%%%%%%%%%%%%%%%%%%%%%%%%

\subsection{Proof of Equation (\ref{equation:diff_wrt_sigma})}
\label{section:demo_diff_wrt_sigma}

Calculating the partial derivative of the normalised forward put price $\hat{P}$ with respect to the volatility of the variance $\sigma$, we get:

\begin{equation*}
    \frac{\partial \hat{P}}{\partial \sigma} =  - \frac{1}{\pi} {\displaystyle \int^{+ \infty}_0 Re \left[ e^{ (i u + \frac{1}{2}) \hat{F}_{t,T}} \ \frac{\partial \phi_{\tau}(u-\frac{i}{2})}{\partial \sigma}\right]\frac{\,du}{u^2+\frac{1}{4}}}
\end{equation*}

Differentiating the characteristic function (\ref{equation:phi}) with respect to $\sigma$ yields:

\begin{equation*}
     \frac{\partial \phi_{\tau}(u)}{\partial \sigma} = \phi_{\tau}(u) \left[
    \frac{\partial A(u)}{\partial \sigma} + v_0 \frac{\partial B(u)}{\partial \sigma} \right]
\end{equation*}

We have:
\begin{equation}
\label{equation:diff_A_sigma_initial_form}
\frac{\partial A(u)}{\partial \sigma} = \frac{-2\ A(u)}{\sigma} + \frac{\kappa \theta}{\sigma^2}\left \{\tau \left(\frac{\partial \beta}{\partial \sigma} - \frac{\partial d}{\partial \sigma} \right) - 2 \ \frac{\partial}{\partial \sigma}\left[ \ln \left (\frac{g e^{-d\tau} -1} {g-1} \right) \right] \right \}
\end{equation}

Differentiating $ \ln \left(\frac{g e^{-d\tau} - 1}{g-1}\right)$ with respect to $\sigma$ requires the calculation of the partial derivative of $\beta$, $d$ and $g$ with respect to $\sigma$ that gives:

\begin{equation}
    \label{equation:diff_beta_sigma}
    \frac{\partial \beta}{\partial \sigma} = - i\ u \ \rho
\end{equation}

\begin{equation}
    \label{equation:diff_d_sigma}
    \frac{\partial d}{\partial \sigma} = \frac{- 1}{d} \left[ i u \rho \beta + 2 \hat{\alpha} \sigma\right]
\end{equation}

\begin{equation}
    \label{equation:diff_g_sigma}
    \frac{\partial g}{\partial \sigma} = \frac{2 i u \rho \ (\beta^2 - d^2) + 4 \beta \hat{\alpha} \sigma}{d \ (\beta + d)^2}
\end{equation}

yielding after further calculation to (\ref{equation:diff_ln_sigma}). Then substituting (\ref{equation:diff_beta_sigma}) and (\ref{equation:diff_d_sigma}) in (\ref{equation:diff_A_sigma_initial_form}) gives (\ref{equation:diff_A_sigma}). Moreover, the partial derivative of $B$ with respect to $\sigma$ is given by (\ref{equation:diff_B_sigma}) and we have:

\begin{equation}
    \label{equation:beta_d_initial_form}
    \frac{\partial}{\partial \sigma} \left(\frac{\beta - d}{\sigma^2} \right) = \frac{\left (\frac{\partial \beta}{\partial \sigma} - \frac{\partial d}{\partial \sigma} \right) \sigma^2 - 2 \sigma (\beta - d)}{\sigma^4}
\end{equation}

By substituting (\ref{equation:diff_beta_sigma}) and (\ref{equation:diff_d_sigma}) in (\ref{equation:beta_d_initial_form}), we obtain (\ref{equation:diff_beta_d_sigma}).  Finally, we have:

\begin{equation}
\label{equation:diff_int_sigma_initial_form}
     \frac{\partial}{\partial \sigma} \left(\frac{1 - e^{-d\tau} }{1 - g e^{-d\tau}} \right) = 
     \frac{\frac{\partial}{\partial \sigma} (1 -  e^{-d\tau}) (1 - g e^{-d\tau}) - (1 - e^{-d\tau}) \frac{\partial}{\partial \sigma} (1 - g e^{-d\tau})}{\left(1 - g e^{-d\tau} \right)^2}
 \end{equation}
that leads after long calculations to (\ref{equation:diff_int_sigma}).

%%%%%%%%%%%%%%%%%%%%%%%%%%%%%%%%%%%%%%%%%%%%%%%%%%%
%%%%%%%%%%%%%%%%%%%%%%%%%%%%%%%%%%%%%%%%%%%%%%%

%\section{Second}

%% bibliography
\bibliographystyle{plain}
%\bibliography{References}

%\printbibliography

% For further references see \href{http://www.overleaf.com}{Something 
% Linky} or go to the next url: \url{http://www.overleaf.com} or open 
% the next file \href{run:./file.txt}{File.txt}

% It's also possible to link directly any word or 
% \hyperlink{thesentence}{any sentence} in your document.

% \begin{figure}[tb]
% \centering
% \includegraphics[width = 0.4\hsize]{./figures/imperial}
% \caption{Imperial College Logo. It's nice blue, and the font is quite stylish. But you can choose a different one if you don't like it.}
% \label{fig:logo}
% \end{figure}

% Figure~\ref{fig:logo} is an example of a figure. 
% Introduce the project and the problem. Provide some context for a reader who is not an expert in this area.

\end{document}